\documentclass[journal,twoside,10pt]{IEEEtran}
\usepackage{amsmath}
\usepackage{mathrsfs}
\usepackage{amsfonts}
\usepackage[numbers,sort&compress]{natbib}
\usepackage{graphicx,color,overpic}
\usepackage{times}
\usepackage{latexsym}
\usepackage{bm}
\usepackage{amssymb}
\usepackage{cases}
\usepackage{array}
\usepackage{fancyhdr}
\usepackage{setspace}

\usepackage[caption=false,font=footnotesize,labelfont=rm,textfont=rm]{subfig}
\usepackage{subfig}
\usepackage{url}
\usepackage{algorithm}
\usepackage{algorithmic}
\usepackage{multirow}
\allowdisplaybreaks[4]

\usepackage{citesort}
\usepackage{epsfig}
\usepackage{fancybox}
\usepackage{textcomp}
\usepackage{psfrag}
\usepackage{makecell}

\newtheorem{theorem}{Theorem}
\newtheorem{lemma}{Lemma}
\newtheorem{remark}{Remark}

    \def\Complex{{\rm\rule[.23ex]{.03em}{1.1ex}\kern-.3em{C}}}

    \newcommand{\be}{\begin{equation}} \newcommand{\ee}{\end{equation}}
    \newcommand{\bea}{\begin{eqnarray}} \newcommand{\eea}{\end{eqnarray}}
    \newcommand{\benum}{\begin{enumerate}} \newcommand{\eenum}{\end{enumerate}}

    \newcommand{\qa}{{\bf a}}
    \newcommand{\qb}{{\bf b}}
    
    \newcommand{\qd}{{\bf d}}

    \newcommand{\qg}{{\bf g}}
    \newcommand{\qh}{{\bf h}}

    \newcommand{\qn}{{\bf n}}

    \newcommand{\qu}{{\bf u}}

    \newcommand{\qx}{{\bf x}}
    \newcommand{\qy}{{\bf y}}

    \newcommand{\qA}{{\bf A}}
    \newcommand{\qB}{{\bf B}}
    \newcommand{\qC}{{\bf C}}
    \newcommand{\qD}{{\bf D}}

    \newcommand{\qG}{{\bf G}}
    
    \newcommand{\qI}{{\bf I}}

    \newcommand{\qQ}{{\bf Q}}
    \newcommand{\qR}{{\bf R}}

    \newcommand{\bbR}{{\mathbb R}}
    \newcommand{\bbC}{{\mathbb C}}

\begin{document}

\title{Low-Complexity Channel Estimation Framework for Non-Square UPA-Assisted XL-MIMO Systems}
\author{Yilong Liu,~\IEEEmembership{Graduate Student Member,~IEEE},
Xi Yang,~\IEEEmembership{Member,~IEEE},
Binggui Zhou,~\IEEEmembership{Member,~IEEE},

Yu Han,~\IEEEmembership{Member,~IEEE},
Ting Liu,~\IEEEmembership{Member,~IEEE},
and Shaodan Ma,~\IEEEmembership{Senior Member,~IEEE}
\thanks{Yilong Liu and Xi Yang are with the School of Information and Electronic Engineering, East China Normal University, Shanghai 200241, China (e-mail: yilongliu@stu.ecnu.edu.cn; xyang@cee.ecnu.edu.cn).}
\thanks{Binggui Zhou is with the Department of Electrical and Electronic Engineering, Imperial College London, SW7 2AZ London, U.K. (e-mail: binggui.zhou@imperial.ac.uk).}
\thanks{Yu Han is with the National Mobile Communications Research Laboratory, Southeast University, Nanjing 210096, China (e-mail: hanyu@seu.edu.cn).}
\thanks{Ting Liu is with the School of Artificial Intelligence, Nanjing University of Information Science and Technology, Nanjing 210044, China (e-mail: liuting@nuist.edu.cn).}
\thanks{Shaodan Ma is with the State Key Laboratory of Internet of Things for Smart City and the Department of Electrical and Computer Engineering, University of Macau, Macao SAR, China (e-mail: shaodanma@um.edu.mo).}
\thanks{A conference version of this work has been accepted in IEEE GLOBECOM 2026 \cite{YLiu-26GLOBECOM}.}
}
\maketitle

\begin{abstract}
Low-complexity channel state information acquisition is crucial for extremely large-scale multiple-input multiple-output (XL-MIMO) systems.
However, practical deployments of non-square uniform planar arrays (UPAs) in hybrid-field environments face prohibitive computational complexity and degraded estimation accuracy due to limited elevation angle-of-arrival (AoA) resolution and deteriorated channel sparsity.
To tackle these challenges, we propose a low-complexity channel estimation framework.
First, an antenna-domain extrapolation scheme synthesizes a virtually enlarged vertical aperture via the spatial correlation among adjacent elements, breaking the elevation resolution limit.
The framework then disentangles the parameter coupling by transforming the two-dimensional joint search into two sequential one-dimensional searches.
Specifically, elevation AoAs are extracted via an extrapolation-enhanced discrete Fourier transform-Newtonized orthogonal matching pursuit (NOMP) algorithm along the virtually enlarged vertical uniform linear array (ULA), while azimuth AoAs, ranges, and gains are acquired utilizing a discrete fractional Fourier transform-NOMP algorithm along a horizontal ULA.
A subspace fitting-driven path matching algorithm pairs these decoupled parameters.
To overcome the accuracy bottleneck of the antenna-domain scheme, a correlation-domain extrapolation scheme is further developed by exploiting the structural properties of the spatial correlation matrix to decouple the near-field quadratic and azimuth phase components, yielding a noise-suppressed virtual array.
Numerical results validate the effectiveness of the proposed framework.
\end{abstract}


\section{Introduction}
Extremely large-scale multiple-input multiple-output (XL-MIMO) has emerged as a pivotal technology for sixth-generation (6G) mobile communication networks to fulfill the stringent performance requirements of future 6G applications \cite{Lu-24COMST}.
In XL-MIMO systems, the enlarged array aperture extends the Rayleigh distance substantially.
Consequently, a portion of the scatterers falls within the near-field region, while the rest remain in the far-field region, constituting a hybrid-field propagation environment \cite{HLei-24TSP}, \cite{WD-22CL}.

To fully exploit the performance gains of XL-MIMO systems, the acquisition of accurate channel state information (CSI) is of vital importance, which is further utilized for effective beamforming and subsequent signal processing \cite{JCao-26TWC}, \cite{SZeng-25WC}.
However, the extremely large number of antennas, combined with the parameter coupling of angles of arrival (AoAs) and ranges in hybrid-field environments, incurs unaffordable computational complexity for channel estimation.
By leveraging the channel sparsity, compressed sensing-based methods circumvent the exhaustive multi-dimensional parameter searches and matrix inversions required by conventional estimators, such as the least squares (LS) and minimum mean square error methods \cite{SYue-24TWC}, \cite{YKang-25TCOM}.
For instance, the authors in \cite{WD-22CL} developed a hybrid-field orthogonal matching pursuit (OMP) algorithm, which sequentially estimates far-field and near-field paths within the angular and polar domains, respectively.
Based on the sparsity of hybrid-field channels in the fractional Fourier domain, \cite{XYang-24WCL} proposed a gridless estimation algorithm.
The authors in \cite{HWang-26TCOMM} demonstrated that the hybrid-field channel admits a block-sparse representation over a specially designed unitary matrix, and subsequently employed block-sparse signal recovery algorithms for channel estimation.

Nevertheless, the aforementioned algorithms assume a uniform linear array (ULA) deployed at the base station (BS).
In practical deployments, uniform planar arrays (UPAs) are widely adopted as the prevailing architecture, owing to their capability to achieve high antenna integration densities within constrained physical spaces \cite{CChen-26TWC}, \cite{3GPP-38.901}.
The extremely large scale of UPAs, combined with near-field spherical wavefronts, induces three-dimensional (3D) parameter coupling among azimuth AoAs, elevation AoAs, and ranges, rendering conventional joint estimation approaches computationally prohibitive \cite{Guo-23TVT}.
Thus, the development of low-complexity channel estimation methods for UPA-assisted XL-MIMO systems is highly warranted.
Although several low-complexity schemes have been developed for near-field UPA-assisted XL-MIMO systems, such as the parameter decoupling-based methods \cite{Guo-23TVT}, \cite{HMiao-25ICC}, \cite{Fazal-25TCOM}, their application to hybrid-field scenarios is restricted by the energy leakage in the spatial domain and the energy spread in the polar domain \cite{XYang-24WCL}, \cite{YLi-25TWC}.
Then, \cite{WYu-23JSTSP} proposed a fixed point network framework to enable low-overhead and adaptive hybrid-field channel estimation tailored for the planar array-of-subarray architecture.
\cite{PZheng-25TWC} investigated channel estimation for reconfigurable intelligent surface-aided XL-MIMO systems under hybrid-field scenarios, and developed a convolutional dictionary learning framework.
Recently, generative artificial intelligence-based schemes, such as flow matching frameworks \cite{ZJiang-26WCNC}, \cite{ZJiang-26arXiv}, have also been investigated to capture complex channel distributions for enhanced estimation performance.
Nevertheless, these channel estimation algorithms predominantly consider square arrays.
In practice, XL-MIMO systems typically deploy non-square UPAs \cite{3GPP-38.901}, \cite{XYang-26JIOT}, \cite{YuH-19TCOM}.
For example, non-square UPAs with larger horizontal apertures but relatively small vertical apertures are employed to efficiently accommodate users densely distributed in the azimuth plane.
However, such a configuration restricts the physical array aperture in the elevation dimension, aggravating the Rayleigh resolution limit and degrading the elevation AoA estimation accuracy.
In the presence of the 3D parameter coupling induced by the extremely large UPA scale and hybrid-field propagation, such inaccurate elevation AoA estimation severely hinders the decoupling of the horizontal and vertical ULAs from the UPA.
Although subspace-based super-resolution techniques, such as the multiple signal classification (MUSIC) algorithm, can be utilized to break the Rayleigh resolution limit, the limited vertical aperture renders their subspace orthogonality highly sensitive to ambient noise and finite snapshots, leading to spectral peak aliasing at moderate-to-low SNRs \cite{Pakrooh-16TSP}.
Moreover, practical hardware constraints dictate that the number of radio frequency (RF) chains is much smaller than the number of UPA antennas, yielding undersampled spatial observations that further complicate accurate channel estimation \cite{Sohrabi-16JSTSP}.

To tackle these challenges, we propose a low-complexity channel estimation framework for non-square UPA-assisted XL-MIMO systems.
By exploiting the spatial correlation among adjacent antenna elements, an antenna-domain extrapolation scheme is firstly developed to synthesize an enlarged vertical aperture.
Building upon this extrapolated array, a channel estimation algorithm is then proposed to resolve the 3D parameter coupling by transforming the prohibitive two-dimensional (2D) joint search into sequential one-dimensional (1D) searches, thereby significantly reducing the computational complexity.
To overcome the accuracy bottleneck of the antenna-domain scheme originating from the recursive noise accumulation, a correlation-domain extrapolation scheme is further proposed to construct a noise-suppressed virtual array based on the statistical CSI.
The primary contributions are summarized as follows.
\begin{itemize}
\item First, we propose a linear prediction-based antenna-domain extrapolation scheme to enlarge the vertical aperture without extra hardware costs, facilitating the decoupling of the UPA into a horizontal and a virtually enlarged vertical ULA.
By partitioning the non-square UPA into column-wise subarrays, we exploit the elevation AoA-governed harmonic structure along each vertical subarray to recursively generate unobserved spatial samples and expand the observation window beyond the physical array boundaries.
Furthermore, we derive the Cram\'er-Rao bound (CRB) for elevation AoA estimation to analyze the trade-off between spatial resolution enhancement and recursive error accumulation, thereby revealing the theoretical limit of the maximum achievable extrapolated aperture.
\item Then, we propose a low-complexity channel estimation algorithm based on the extrapolated array by transforming the coupled 2D estimation problem into two sequential 1D estimation problems along the virtually enlarged vertical ULA and the horizontal ULA.
An extrapolation-enhanced discrete Fourier transform-Newtonized orthogonal matching pursuit (Ex-DFT-NOMP) algorithm is developed to extract elevation AoAs while mitigating off-grid errors along the virtually enlarged vertical ULA.
Subsequently, by employing the sparsity in the fractional Fourier domain, a hybrid-field discrete fractional Fourier transform (DFrFT)-NOMP algorithm is executed to acquire the azimuth AoAs, ranges, and complex gains along a horizontal ULA.
A subspace fitting-driven path matching algorithm is further formulated to associate these decoupled parameters.
\item Finally, we develop a correlation-domain extrapolation scheme to achieve higher estimation accuracy.
This scheme exploits the structural properties of the spatial correlation matrix to cancel the near-field quadratic and azimuth-dependent phase components, thereby isolating the elevation-dependent harmonic component and constructing a noise-suppressed equivalent virtual ULA for subsequent aperture extrapolation and channel estimation.
To mitigate the central noise impulse generated during this matrix construction, we propose a bidirectional linear prediction-based interpolation strategy.
We also derive the CRB for this scheme to characterize the theoretical elevation AoA estimation limit of the correlation-domain extrapolation scheme.
Numerical results validate the effectiveness and superiority of the proposed framework.
\end{itemize}

The remainder of this paper is organized as follows.
Section\,\ref{s:Model} introduces the considered non-square UPA-assisted XL-MIMO system, and formulates the hybrid-field channel estimation problem.
In Section\,\ref{s:APEXante}, the antenna-domain extrapolation scheme is proposed, and we propose a low-complexity channel estimation algorithm in Section\,\ref{s:channel_est}.
Then, the correlation-domain extrapolation scheme is proposed in Section\,\ref{s:APEXcorr}.
Section\,\ref{s:Simulations} presents the numerical results, and finally we conclude this paper in Section\,\ref{s:Conclusion}.

\textit{Notations:} Scalars, vectors, and matrices are denoted by non-boldface, boldface lowercase, and boldface uppercase letters, respectively.
Superscripts $(\cdot)^T$, $(\cdot)^H$, and $(\cdot)^{-1}$ denote the transpose, Hermitian transpose, and inverse, respectively.
$\mathbb{E}(\cdot)$, $\mathrm{tr}(\cdot)$, $\mathrm{diag}(\cdot)$, $\arg(\cdot)$, $\lfloor \cdot \rfloor$, $\mathcal{F}^{(q)}\{\cdot\}$, and $\mathcal{F}^{(1)}\{\cdot\}$ are the statistical expectation, trace, diagonal, phase, floor, $q$-order DFrFT, and DFT operations, respectively.
$\bbC^{M\times N}$ and $\bbR^{M\times N}$ denote the spaces of $M \times N$ complex and real matrices, respectively.
Hadamard product, Kronecker product, and magnitude operator are denoted by $\odot$, $\otimes$, and $|\cdot|$, respectively.
$[\qA]_{i,j}$ denotes the $(i,j)$-th element of the matrix $\qA$.
$[\qa]_i$, $[\qa]_{\mathcal{S}}$, and $\|\qa\|$ denote the $i$-th element, the subvector, and the Euclidean norm of the vector $\qa$, respectively.
$\Re(\cdot)$ and $\Im(\cdot)$ denote the real and imaginary parts, respectively.
$\mathbf{0}_{M \times 1}$ denotes an $M \times 1$ all-zero vector.
$\mathcal{C}\mathcal{N}(\mu, \sigma^2)$ denotes a complex Gaussian distribution with mean $\mu$ and variance $\sigma^2$.
$\mathcal{U}(a,b)$ represents the uniform distribution over the interval $[a, b]$.

\section{System Model and Problem Formulation}\label{s:Model}
\subsection{System Model}
\begin{figure}[htbp]
       \centering
       \includegraphics[width=0.46\textwidth]{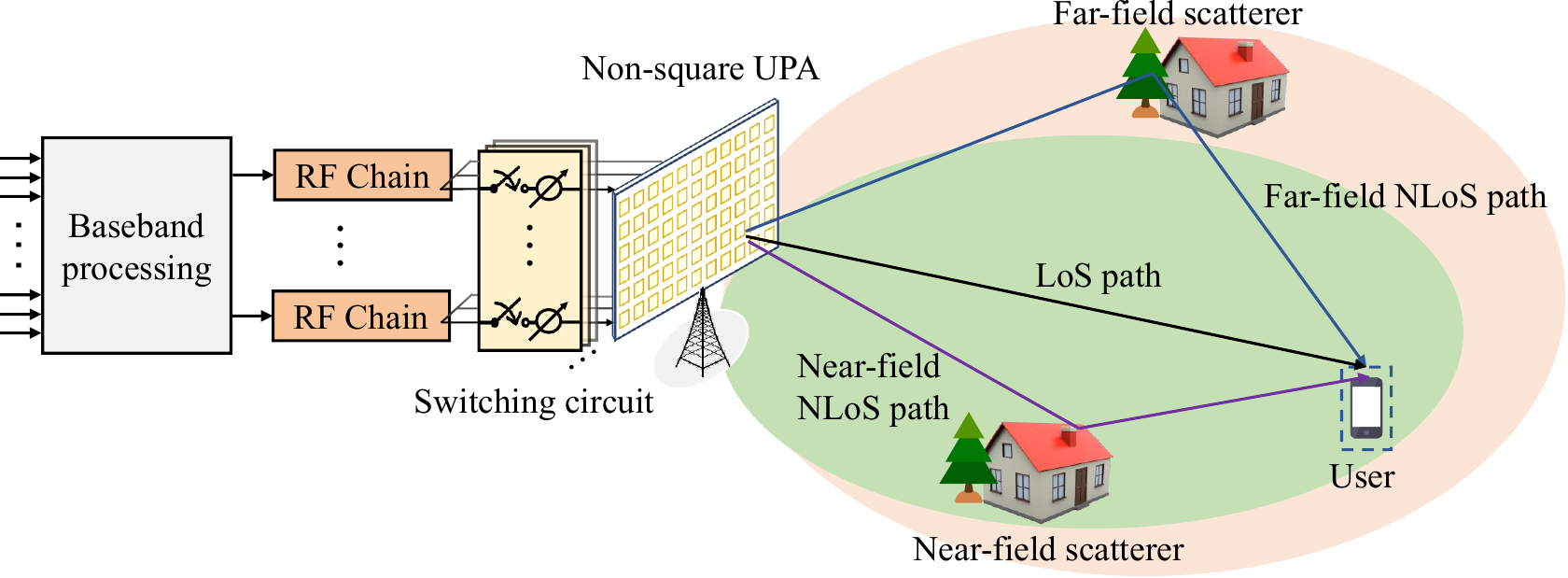}
       \caption{The non-square UPA-assisted hybrid-field XL-MIMO system.}\label{fig:sys_model}
\end{figure}

As depicted in Fig.\,\ref{fig:sys_model}, we consider an XL-MIMO system featuring an antenna selection architecture.\footnote{The proposed scheme can be readily extended to alternative architectures, such as fully-connected and subarray-based hybrid structures.}
The BS, equipped with a non-square UPA consisting of $M=M_1 \times M_2$ antennas, serves a single-antenna user, where $M_1$ and $M_2$ denote the numbers of antennas along the vertical and horizontal directions, respectively, and $M_1 < M_2$.
Under the considered antenna selection architecture, the $M$ antennas are connected to $N$ ($N \ll M$) RF chains via a switching network, allowing at most $N$ antennas to be activated simultaneously.
Without loss of generality, a set of $N$ antennas is assumed to be activated.
Compared with a fully digital architecture, the considered architecture substantially reduces the required number of RF chains and the associated hardware cost and power consumption.
Given the large array aperture at the BS, the propagation environment is characterized by the hybrid-field scattering, i.e., there are near-field and far-field scatterers in the propagation environment simultaneously \cite{WD-22CL}.
With $s = 1$ denoting the uplink pilot symbol transmitted by the user,\footnote{Although a single-user scenario with $s=1$ is considered for brevity, the proposed framework can be extended to multi-user scenarios by assigning orthogonal pilot sequences to different users \cite{CD-22TCOMM}, \cite{Marzetta-10TWC}.
After pilot-matched filtering, the observation associated with each user follows the same form as \eqref{eq:y1}, allowing the proposed framework to be applied on a per-user basis.} the received signal $\qy_1 \in \bbC^{M \times 1}$ at the BS is modeled as
\begin{align}\label{eq:y1}
       \qy_1 = \qu \odot (\qh_1 s + \qn_1),
\end{align}
where $\qu \in \{0,1\}^{M \times 1}$ is the antenna selection indicator vector, with its $N$ non-zero elements corresponding to the activated antennas,
$\qn_1 \sim \mathcal{C}\mathcal{N}(\mathbf{0}_{M \times 1}, \sigma^2_\mathrm{n}\qI_M)$ represents the received zero-mean complex additive white Gaussian noise (AWGN) vector,
and $\qh_1 \in \bbC^{M \times 1}$ denotes the channel vector between the user and the BS, which is formulated as
\begin{align}\label{eq:h1}
       \qh_1 = \sqrt{\frac{M}{K}} \left(\sum_{k=1}^{K_\mathrm{f}} \beta_k \qa(\theta_k, \phi_k) +\sum_{k=K_\mathrm{f}+1}^{K} \beta_k \qb(\theta_k, \phi_k, r_k)\right),
\end{align}
where $\beta_k \sim \mathcal{C}\mathcal{N}(0, \sigma^2_k)$ is the complex gain of the $k$-th path. 
The parameters $K_\mathrm{f}$ and $K_\mathrm{n}$ denote the numbers of far-field and near-field paths, respectively, and we have $K = K_\mathrm{f} + K_\mathrm{n}$.
$\theta_k$, $\phi_k$, and $r_k$ represent the elevation AoA, azimuth AoA, and range for the $k$-th path, respectively.
$\qa(\theta, \phi) \in \bbC^{M \times 1}$ and $\qb(\theta, \phi, r) \in \bbC^{M \times 1}$ are the far-field and the near-field steering vectors, respectively.
By omitting the cross-terms \cite{YPan-23JSTSP}, the steering vectors can be approximated as $\qa(\theta, \phi) \approx \qa_1(\theta) \otimes \qa_2(\theta, \phi)$ and $\qb(\theta, \phi, r) \approx \qa_1(\theta) \otimes \qb_2(\theta, \phi, r)$ with
\begin{subequations}
\begin{align}
       [\qa_1(\theta)]_i =& \frac{1}{\sqrt{M_1}} e^{j \frac{2\pi a_{\mathrm{v},i}d \sin \theta}{\lambda}},\\
       [\qa_2(\theta, \phi)]_i =& \frac{1}{\sqrt{M_2}} e^{j \frac{2\pi a_{\mathrm{h},i}d \cos \theta \cos \phi}{\lambda}},\\
       [\qb_2(\theta, \phi, r)]_i =& \frac{1}{\sqrt{M_2}} e^{j \frac{2\pi}{\lambda}(a_{\mathrm{h},i}d \cos \theta \cos \phi - a_{\mathrm{h},i}^2\frac{d^2(1-\cos^2 \theta \cos^2 \phi)}{2r})},
\end{align}
\end{subequations}
where $\lambda$ is the carrier wavelength, $d$ denotes the antenna spacing at the BS, $a_{\mathrm{v},i} \in \mathcal{M}_1$ and $a_{\mathrm{h},i} \in \mathcal{M}_2$ denote the vertical and horizontal antenna indices, with the index sets defined as $\mathcal{M}_1 \triangleq \{-\frac{M_1-1}{2}, \dots, \frac{M_1-1}{2}\}$ and $\mathcal{M}_2 \triangleq \{-\frac{M_2-1}{2}, \dots, \frac{M_2-1}{2}\}$, respectively.
Note that $\qb(\theta, \phi, r)$ adopts the Fresnel approximation in the horizontal direction and the planar wavefront assumption in the vertical direction.
For a non-square UPA, this model applies to the resulting anisotropic near-field region, where the user is in the near field of the horizontal (long) axis but in the far field of the vertical (short) axis.
At a typical array aspect ratio of $\gamma = 4$, this region occupies almost the entire near-field region \cite{YLiu-26arXiv-analysis}.

To facilitate subsequent analysis, we introduce a unified representation for both far-field and near-field steering vectors, denoted by $\tilde{\qb}(\theta, \vartheta_{\mathrm{h}}, \varphi_{\mathrm{h}}) = \qa_1(\theta) \otimes \tilde{\qb}_2(\vartheta_{\mathrm{h}}, \varphi_{\mathrm{h}})$ with
\begin{align}
      [\tilde{\qb}_2(\vartheta_{\mathrm{h}}, \varphi_{\mathrm{h}})]_i = \frac{1}{\sqrt{M_2}} e^{j a_{\mathrm{h},i}\vartheta_{\mathrm{h}}}e^{-j a_{\mathrm{h},i}^2\varphi_{\mathrm{h}}},
\end{align}
where $\vartheta_{\mathrm{h}} \triangleq \frac{2\pi d \cos \theta \cos \phi}{\lambda}$ and $\varphi_{\mathrm{h}} \triangleq \frac{\pi d^2 (1-\cos^2 \theta \cos^2 \phi)}{\lambda r}$.
For the far-field paths, we have $\qa_2(\theta, \phi) = \tilde{\qb}_2(\vartheta_{\mathrm{h}}, 0)$, and thus $\qh_1$ in \eqref{eq:h1} can be rewritten as
\begin{align}
       \qh_1 = \sqrt{\frac{M}{K}} \sum_{k = 1}^{K} \beta_k \tilde{\qb}(\theta_k, \vartheta_{\mathrm{h},k}, \varphi_{\mathrm{h},k}).
\end{align}

\subsection{Problem Formulation}
In XL-MIMO systems, acquiring the full-dimensional CSI, i.e., $\qh_1$, is imperative to fully harness the high spatial resolution and extensive degrees of freedom.
For instance, such full-dimensional CSI serves as a fundamental prerequisite for designing the precoding matrix, selecting the optimal active antenna subset, and performing efficient resource allocation \cite{YChen-24TWC}, \cite{Sharifi-22TCOM}.
Therefore, in this paper, our objective is to estimate the full-dimensional channel vector $\qh_1$ from the undersampled observations $\qy_1$.
However, practical deployments of non-square UPAs severely limit the vertical physical array aperture.
For example, the representative $8 \times 32$ UPA in 3GPP TR 38.901 has a vertical aperture approximately one quarter of its horizontal aperture \cite{3GPP-38.901}.
According to array theory, the angular resolution along each array dimension is inversely proportional to the corresponding aperture, and thus the elevation AoA resolution is approximately four times coarser than the azimuth AoA resolution.
Consequently, directly decomposing the 2D channel estimation into two independent 1D procedures will inevitably incur severe performance loss due to the inadequate vertical angular accuracy.
On the other hand, executing a joint 2D parameter search over the entire UPA will introduce unaffordable computational complexity.
These challenges are further exacerbated by the hybrid-field propagation, which degrades both spatial and polar sparsity \cite{YLi-25TWC}.

Fortunately, this limited vertical aperture concurrently yields a minimal vertical Rayleigh distance, ensuring that the signal wavefronts along the elevation dimension remain predominantly planar.
Such a planar wave characteristic endows the received signals with a harmonic structure (i.e., a superposition of complex sinusoids).
Motivated by this property, we first develop a linear prediction-based aperture extrapolation scheme in the antenna domain in the following, which exploits the spatial correlation among adjacent antennas to synthesize an enlarged vertical aperture, thereby alleviating the elevation resolution bottleneck.

\section{Antenna-Domain Aperture Extrapolation}\label{s:APEXante}
By leveraging the underlying harmonic structure of spatial array signals, we propose a linear prediction-based extrapolation scheme in the vertical antenna domain in this section.
Note that this approach enlarges the spatial observation window and enhances the elevation AoA resolution without physically enlarging the UPA aperture.
Therefore, the corresponding effective Rayleigh distance remains unchanged.

To this end, we first partition the entire large UPA into multiple column-wise subarrays, where each vertical column constitutes an independent subarray, and the received signal of the BS corresponding to the $m$-th subarray is given by\footnote{We assume that the $m$-th vertical column and the central horizontal row of the UPA are fully activated to obtain complete observations of these two subarrays for the subsequent parameter estimation. If these specific antennas are only partially activated, the signals of the unactivated elements can be recovered via the interpolation strategy detailed in Section\,\ref{s:APEXcorr}.}
\begin{align}
       \qy_{\mathrm{v},m} = \qh_{\mathrm{v},m} + \qn_{\mathrm{v},m},
\end{align}
where $\qy_{\mathrm{v},m} \in \mathbb{C}^{M_1 \times 1}$, $\qn_{\mathrm{v},m}$ denotes the received complex AWGN corresponding to the $m$-th subarray, i.e., $\qn_{\mathrm{v},m} \sim \mathcal{C}\mathcal{N}(\mathbf{0}_{M_1 \times 1}, \sigma^2_\mathrm{n}\qI_{M_1})$, and $\qh_{\mathrm{v},m} \in \bbC^{M_1 \times 1}$ is the channel vector corresponding to the $m$-th subarray, as follows
\begin{align}\label{eq:h_vm}
       \qh_{\mathrm{v},m} = \sqrt{\frac{M_1}{K}} \sum_{k = 1}^{K} \beta_k [\tilde{\qb}_2(\vartheta_{\mathrm{h},k}, \varphi_{\mathrm{h},k})]_m \qa_1(\theta_k).
\end{align}
Note that $\qh_{\mathrm{v},m}$ in \eqref{eq:h_vm} can be formulated as a linear combination of $K$ complex exponential signals with linear phases determined by the elevation AoAs.
Then, to virtually enlarge the vertical array aperture based on the extrapolation of the received signal $\qy_{\mathrm{v},m}$ in the antenna domain, we introduce the following lemma.

\begin{lemma}\label{lem:extrapolation}
       Any linear combination of $K$ complex exponential signals with linear phases necessarily satisfies a $K$-th order constant-coefficient linear homogeneous difference equation, i.e., there exists a unique set of coefficients $\{\bar{c}_1,\dots,\bar{c}_K\}$ such that $[\qh_{\mathrm{v},m}]_i + \bar{c}_1[\qh_{\mathrm{v},m}]_{i-1} + \dots + \bar{c}_K[\qh_{\mathrm{v},m}]_{i-K}=0$.
\end{lemma}

\begin{proof}
See Appendix\,\ref{proof:lem:extrapolation}.
\end{proof}

Thus, each subsequent sample of the signal sequence in the antenna domain can be expressed as a linear combination of the previous $K$ samples, and the aperture extrapolation problem can be formulated as a linear prediction problem.
In the $z$ domain, this is equivalent to constructing an all-pole filter defined by the prediction coefficients with its transfer function given by $\bar{H}(z) = \frac{1}{1 - \sum_{k=1}^{K}(-\bar{c}_k)z^{-k}}$.
Based on the maximum entropy estimation theory \cite{burg1975maximum}, the power response of this filter, i.e., $|\bar{H}(z)|^2$, characterizes the autoregressive power spectral density of the observed signal, implying that the aperture extrapolation can be achieved by recursively applying these filter coefficients to predict subsequent spatial samples.
However, since $\qy_{\mathrm{v},m}$ contains additive noise, setting the prediction order strictly equal to the number of paths may lead to noise-induced pole perturbations.
Therefore, a higher prediction order is required to exploit additional degrees of freedom, as follows
\begin{align}\label{eq:filter}
       H(z) = \frac{1}{1 - \sum_{p=1}^{P}(-c_p)z^{-p}},\,P>K,
\end{align}
where $c_p$ is the linear prediction coefficient and will be determined later.
Let $y_{\mathrm{v},m,i} = [\qy_{\mathrm{v},m}]_i$ denote the $i$-th element of the original observation sequence, and the forward extrapolated samples can be recursively generated by $y_{\mathrm{v},m,i} = \sum_{p=1}^{P} (-c_p) y_{\mathrm{v},m,i-p}$ for $i > M_1$.
Notice that given the recursive nature of such aperture extrapolation, an unstable filter with poles outside the unit circle will cause noise-induced errors to amplify exponentially, leading to severe data divergence.
To circumvent this issue, the Burg algorithm \cite{burg1975maximum} is employed to compute the linear prediction coefficients, as it guarantees that all filter poles remain within the unit circle \cite{Swingler-89TSP}.
The computational complexity introduced by the Burg algorithm is $\mathcal{O}(M_1 P)$, which maintains the low-complexity advantage of the proposed framework.

To further extend the aperture, the backward extrapolation is also performed after the forward process.
Note that the linear prediction coefficients of stationary random signals exhibit conjugate symmetry in the forward and backward directions.
Therefore, the backward extrapolation can be achieved by recursively applying the all-pole filter in \eqref{eq:filter} along the reverse array axis.
Specifically, the backward extrapolated samples are sequentially generated by $y_{\mathrm{v},m,i} = \sum_{p=1}^{P} (-c_p^*) y_{\mathrm{v},m,i+p}$ for $i < 1$ in descending order of the index $i$.
Finally, these newly generated samples are concatenated with the original observation sequence to construct the fully extrapolated signal vector.
We provide the detailed expression of the virtually received signal from the extrapolated vertical array in the following theorem.

\begin{theorem}\label{the:extra_signal_antenna_domain}
       In the antenna domain, the extrapolated signal $\bar{\qy}_{\mathrm{v},m} \in \bbC^{\bar{M}_1 \times 1}$ can be formulated as
       \begin{align}\label{eq:yvm_bar}
              \bar{\qy}_{\mathrm{v},m} = \sum_{k = 1}^{K} \bar{\beta}_{k,m} \bar{\qa}_1(\theta_k) + \bar{\qn}_{\mathrm{v},m},
       \end{align}
       where $\bar{\beta}_{k,m} \triangleq \sqrt{\frac{\bar{M}_1}{K}} \beta_k [\tilde{\qb}_2(\vartheta_{\mathrm{h},k}, \varphi_{\mathrm{h},k})]_m$, $\bar{M}_1 \triangleq 2\lfloor \xi_1 M_1 \rfloor+M_1$, and $\xi_1$ denotes the extrapolation factor, i.e., $\lfloor \xi_1 M_1 \rfloor$ represents the length of the forward/backward extrapolated sequence.
       $\bar{\qn}_{\mathrm{v},m} \triangleq [\tilde{\qn}_{\mathrm{v},m}^T,\qn_{\mathrm{v},m}^T,\breve{\qn}_{\mathrm{v},m}^T]^T$ denotes the virtual received noise for the extrapolated array, $\breve{\qn}_{\mathrm{v},m} \sim \mathcal{C}\mathcal{N}(\mathbf{0}_{\lfloor \xi_1 M_1 \rfloor \times 1}, \breve{\qR}_{\mathrm{v},m})$ and $\tilde{\qn}_{\mathrm{v},m} \sim \mathcal{C}\mathcal{N}(\mathbf{0}_{\lfloor \xi_1 M_1 \rfloor \times 1}, \tilde{\qR}_{\mathrm{v},m})$ represent the noise induced by the forward and backward extrapolation processes with $\breve{\qR}_{\mathrm{v},m} $ defined in \eqref{eq:R_v_breve} and $\tilde{\qR}_{\mathrm{v},m} $ defined in \eqref{eq:R_v_tilde}, respectively.
       $\bar{\qa}_1(\theta_k) \in \bbC^{\bar{M}_1 \times 1}$ denotes the steering vector corresponding to the $k$-th path for the extrapolated array, which can be written as
       \begin{align}
              \bar{\qa}_1(\theta_k) =  \frac{1}{\sqrt{\bar{M}_1}} \bigg[ e^{j\frac{2\pi}{\lambda}(-\frac{\bar{M}_1-1}{2})d\sin \theta_k}, \dots,e^{j\frac{2\pi}{\lambda}\frac{\bar{M}_1-1}{2}d\sin \theta_k}\bigg]^T\!.
       \end{align}
\end{theorem}

\begin{proof}
See Appendix\,\ref{proof:the:extra_signal_antenna_domain}.
\end{proof}

Since the half-power mainlobe width is inversely proportional to the array aperture, extending the vertical aperture from $(M_1-1)d$ to $(\bar{M}_1-1)d$ narrows the mainlobe and thereby improves the elevation AoA resolution.
Based on the analytical model for the antenna-domain extrapolated signals established in {\it Theorem\,\ref{the:extra_signal_antenna_domain}}, we derive the CRB in {\it Theorem\,\ref{the:CRBante}} to evaluate the theoretical performance limit for the elevation AoA estimation.

\begin{theorem}\label{the:CRBante}
       Based on the extrapolated signal $\bar{\qy}_{\mathrm{v},m}$ in \eqref{eq:yvm_bar}, the equivalent Fisher information matrix (FIM) for the elevation AoA vector $\boldsymbol{\theta} \triangleq [\theta_1,\dots,\theta_K]^T$, with the complex path gain vector $\boldsymbol{\beta} \triangleq [\beta_1,\dots,\beta_K]^T$ treated as an unknown deterministic nuisance parameter, is given by
       \begin{align}
              \mathbf{F}_{\mathrm{eq,ante}} = 2 \Re \left\{ \boldsymbol{\Lambda}_{\beta}^H \dot{\mathbf{A}}_{\theta}^H \mathbf{R}_{\bar{n},\mathrm{v},m}^{-1} \mathbf{P}_{\mathrm{A}} \dot{\mathbf{A}}_{\theta} \boldsymbol{\Lambda}_{\beta} \right\},
       \end{align}
       where $\boldsymbol{\Lambda}_{\beta} \triangleq \mathrm{diag}(\boldsymbol{\beta})$, $\dot{\mathbf{A}}_{\theta} \triangleq \frac{\partial \bar{\mathbf{A}}}{\partial \boldsymbol{\theta}^T}$ with $\bar{\mathbf{A}}$ denoting the extrapolated array manifold matrix whose $k$-th column is $\bar{\mathbf{a}}_1(\theta_k)$, $\mathbf{P}_{\mathrm{A}} = \mathbf{I}_{\bar{M}_1} - \bar{\mathbf{A}}(\bar{\mathbf{A}}^H \mathbf{R}_{\bar{n},\mathrm{v},m}^{-1} \bar{\mathbf{A}})^{-1} \bar{\mathbf{A}}^H \mathbf{R}_{\bar{n},\mathrm{v},m}^{-1}$,
       and $\mathbf{R}_{\bar{n},\mathrm{v},m}$ is the correlation matrix for $\bar{\qn}_{\mathrm{v},m}$, defined as \eqref{eq:R_n_ante}.
       Thus, the CRB for $\theta_k$ is
       \begin{align}
             \mathrm{CRB}_{\mathrm{ante}}(\theta_k) = [\mathbf{F}_{\mathrm{eq,ante}}^{-1}]_{k,k}. 
       \end{align}
\end{theorem}

\begin{proof}
See Appendix\,\ref{proof:the:CRBante}.
\end{proof}

\begin{remark}
       The CRB derived in {\it Theorem\,\ref{the:CRBante}} reveals that the elevation AoA estimation accuracy exhibits a non-monotonic evolution with the extrapolated aperture $\bar{M}_1$.
\end{remark}

This non-monotonic phenomenon is governed by the trade-off between the spatial resolution gain and the recursive prediction error accumulation.
On one hand, enlarging $\bar{M}_1$ amplifies the spatial indices of the edge antennas, which increases the norm of $\dot{\mathbf{A}}_{\theta}$ and elevates the Fisher information.
Concurrently, the extended dimensionality of the manifold matrix $\bar{\mathbf{A}}$ drives the steering vectors of distinct paths toward spatial orthogonality, mitigating the multipath interference coupling.
On the other hand, due to the recursive nature of linear prediction, the single-step prediction error accumulates and amplifies rapidly within the correlation matrix $\mathbf{R}_{\bar{n},\mathrm{v},m}$.
Once the extrapolation steps exceed a threshold, the exponential error penalty overwhelms the resolution gain, thereby dictating an upper bound on the aperture expansion.
As verified in Fig.\,\ref{fig:NMSE-CRBaoa_SNR}(a), this trade-off restricts the maximum achievable extrapolated aperture of the antenna-domain scheme to approximately $2.5$ to $4$ times the original physical size.

\begin{remark}
The proposed aperture extrapolation scheme preserves the linear prediction relation under common RF-chain phase noise and can accommodate missing observations caused by temporary antenna unavailability or faults through interpolation.
First, in practical systems employing a common RF local oscillator, all antenna observations undergo the same instantaneous phase-noise rotation \cite{Pitarokoilis-15TWC}.
Their relative phases therefore remain unchanged, and the rotated sequence satisfies the original linear prediction relation with the same prediction coefficients.
Hence, common RF-chain phase noise does not affect the proposed aperture extrapolation scheme.
Second, if an individual antenna becomes temporarily unavailable or faulty, its missing observation can be reconstructed before Burg-based extrapolation through the bidirectional linear-prediction interpolation strategy detailed in Section\,\ref{s:APEXcorr}.
\end{remark}

\section{Low-Complexity Channel Estimation}\label{s:channel_est}
In this section, we propose a low-complexity channel estimation algorithm for non-square UPA-assisted hybrid-field XL-MIMO systems.
First, the elevation AoAs are estimated from the extrapolated arrays via the proposed Ex-DFT-NOMP algorithm.
Subsequently, a hybrid-field DFrFT-NOMP algorithm is provided to acquire the azimuth AoAs, ranges, and the corresponding complex gains from the central horizontal row of the UPA.
Finally, a subspace fitting-driven path matching algorithm is conducted to associate the decoupled multipath parameters by aligning them with the low-dimensional channel estimate.

\subsection{Estimation for the Elevation AoAs}
By searching for the peak with maximum amplitude from the DFT spectrum of the residual vector for $\bar{\qy}_{\mathrm{v},m}$, which is initialized as $\bar{\qy}_{\mathrm{v},m,\mathrm{r}} \triangleq [\bar{\qy}_{\mathrm{v},m}^T, \mathbf{0}_{\bar{M}_1(\nu_1-1) \times 1}^T]^T$, we obtain the coarse estimates for the elevation AoA, i.e., $\hat{\theta}_{\max}$, and the corresponding complex gain, i.e., $\hat{\beta}_{1,\max}$, as follows
\begin{subequations}
\begin{align}
       \hat{\theta}_{\max} =& \arcsin \left(\frac{\bar{i} \lambda}{\bar{M}_1\nu_1d} \right), \label{eq:coarseEst_eAoA}\\
       \hat{\beta}_{1,\max} =& \frac{\{\bar{\qa}_1(\hat{\theta}_{\max})\}^H\tilde{\qy}_{\mathrm{v},m,\mathrm{r}}}{\|\bar{\qa}_1(\hat{\theta}_{\max})\|^2},\label{eq:coarseEst_beta}
\end{align}
\end{subequations}
with $\bar{i} = \arg \mathop{\max}_{i} \{ |[\qg_{\mathrm{v},m}]_i|\, \big| \qg_{\mathrm{v},m} \triangleq \mathcal{F}^{(1)}\{\bar{\qy}_{\mathrm{v},m,\mathrm{r}}\} \}$,
where $\nu_1$ denotes the oversampling factor and $\tilde{\qy}_{\mathrm{v},m,\mathrm{r}}$ consists of the first $\bar{M}_1$ elements of $\bar{\qy}_{\mathrm{v},m,\mathrm{r}}$.
In addition, the residual vector should be updated as
\begin{align}\label{eq:yvm_r_dot}
       \dot{\qy}_{\mathrm{v},m,\mathrm{r}} = \bar{\qy}_{\mathrm{v},m,\mathrm{r}} - \hat{\beta}_{1,\max} \dot{\qa}_1(\hat{\theta}_{\max}),
\end{align}
where $\dot{\qa}_1(\hat{\theta}_{\max}) \triangleq [\{\bar{\qa}_1(\hat{\theta}_{\max})\}^T, \mathbf{0}_{\bar{M}_1(\nu_1-1) \times 1}^T]^T$.

To alleviate the off-the-grid effect from the aforementioned peak detection procedure, we propose to apply the Newton-based refinements for $\hat{\theta}_{\max}$ and $\hat{\beta}_{1,\max}$.
Specifically, $R_\mathrm{L}$ local refinements are executed with the objective of minimizing the residual energy $E_1(\theta) = \|\dot{\qy}_{\mathrm{v},m,\mathrm{r}} - \hat{\beta}_{1,\max} \dot{\qa}_1(\theta)\|^2$.
For the $(t+1)$-th iteration, the updates are given by
\begin{subequations}\label{eq:local_refine}
\begin{align}
       \hat{\beta}_{1,\max}^{(t+1)} =& \frac{\{\dot{\qa}_1(\hat{\theta}_{\max}^{(t)})\}^H }{\|\dot{\qa}_1(\hat{\theta}_{\max}^{(t)})\|^2} \left\{ \dot{\qy}_{\mathrm{v},m,\mathrm{r}}^{(t)} + \hat{\beta}_{1,\max}^{(t)} \dot{\qa}_1(\hat{\theta}_{\max}^{(t)}) \right\},\\
       \dot{\qy}_{\mathrm{v},m,\mathrm{r}}^{(t+1)} =& \dot{\qy}_{\mathrm{v},m,\mathrm{r}}^{(t)} - (\hat{\beta}_{1,\max}^{(t+1)} - \hat{\beta}_{1,\max}^{(t)}) \dot{\qa}_1(\hat{\theta}_{\max}^{(t)}),\\
       \hat{\theta}_{\max}^{(t+1)} =& \hat{\theta}_{\max}^{(t)} - \frac{E'_1(\hat{\theta}_{\max}^{(t)})}{E''_1(\hat{\theta}_{\max}^{(t)})},
\end{align}
\end{subequations}
where $E'_1(\theta)$ and $E''_1(\theta)$ denote the first-order and second-order derivatives of $E_1(\theta)$, respectively.
Note that the updates of $\hat{\beta}_{1,\max}^{(t+1)}$ and $\hat{\theta}_{\max}^{(t+1)}$ are adopted only when $E''_1(\hat{\theta}_{\max}^{(t)}) > 0$ and $\|\dot{\qy}_{\mathrm{v},m,\mathrm{r}}^{(t+1)}\|^2 \leq \|\dot{\qy}_{\mathrm{v},m,\mathrm{r}}^{(t)}\|^2$, which ensures the effective reduction of residual energy and guarantees the convergence of the refinement procedure.

After local refinement, we update $\mathcal{P}_1 = \{\mathcal{P}_1, \hat{\theta}_{\max}^{(t+1)}\}$ and $\mathcal{Q}_1 = \{\mathcal{Q}_1, \hat{\beta}_{1,\max}^{(t+1)}\}$, and implement global refinements.
In this step, each $\hat{\theta}_1 \in \mathcal{P}_1$ and $\hat{\beta}_1 \in \mathcal{Q}_1$ are iteratively refined by performing local refinement for $R_G$ times.
To further reduce the energy of the residual vector $\dot{\qy}_{\mathrm{v},m,\mathrm{r}}$, $\hat{\boldsymbol{\beta}}_1 \triangleq [\hat{\beta}_{1,1},\dots,\hat{\beta}_{1,|\mathcal{Q}_1|}]^T$ is updated based on the LS method, as follows
\begin{align}\label{eq:beta_hat}
       \hat{\boldsymbol{\beta}}_1 = (\qA_1^H\qA_1 + \sigma_0^2\qI_{|\mathcal{P}_1|})^{-1}\qA_1^H \bar{\qy}_{\mathrm{v},m},
\end{align}
where $\qA_1 \triangleq [\bar{\qa}_1(\hat{\theta}_{1,1}), \dots, \bar{\qa}_1(\hat{\theta}_{1,|\mathcal{P}_1|})]$, $\sigma_0^2\qI$ is the regularization term that ensures the full rank of $\qA_1^H\qA_1 + \sigma_0^2\qI_{|\mathcal{P}_1|}$, and we set $\sigma_0^2 = 10^{-5}$ in the simulation.
Then, we update $\dot{\qy}_{\mathrm{v},m,\mathrm{r}} = \bar{\qy}_{\mathrm{v},m} - \qA_1 \hat{\boldsymbol{\beta}}_1$.

The algorithm terminates when the residual energy falls below the expected effective error energy, i.e., $\|\dot{\qy}_{\mathrm{v},m,\mathrm{r}}\|^2 < \tau_1$, where $\tau_1$ denotes the threshold, which is given by
\begin{align}
       \tau_1 =& \mathbb{E} \left\{\|\bar{\qn}_{\mathrm{v},m}\|^2 \right\} \notag\\
       =& \mathrm{tr}(\breve{\qR}_{\mathrm{v},m}) + \mathrm{tr}(\tilde{\qR}_{\mathrm{v},m}) + \sigma^2_\mathrm{n} M_1.
\end{align}
The covariance terms in $\tau_1$ account for both the recursively amplified AWGN and the prediction errors arising from parameter mismatch during the forward and backward extrapolations.
This enables reliable path extraction until the residual energy approaches the effective error level (i.e., the expected aggregate energy of the original AWGN and the forward/backward extrapolation errors), and any paths remaining at termination have only a marginal impact on channel reconstruction.
Then, the proposed Ex-DFT-NOMP algorithm is summarized in \textbf{Algorithm$\,$\ref{alg:DFT_NOMP}}.

\begin{algorithm}[!ht]
       \caption{Ex-DFT-NOMP Algorithm}
       \begin{algorithmic} [1]\label{alg:DFT_NOMP}
       \STATE \textbf{Initialize:} $\bar{\qy}_{\mathrm{v},m,\mathrm{r}} \triangleq [\bar{\qy}_{\mathrm{v},m}^T, \mathbf{0}_{\bar{M}_1(\nu_1-1) \times 1}^T]^T$, and $\dot{\qy}_{\mathrm{v},m,\mathrm{r}} = \bar{\qy}_{\mathrm{v},m,\mathrm{r}}$.
       \WHILE{$\|\dot{\qy}_{\mathrm{v},m,\mathrm{r}}\|^2 \geq \tau_1$}
       \STATE Obtain $\hat{\theta}_{\max}$ and $\hat{\beta}_{1,\max}$ based on \eqref{eq:coarseEst_eAoA} and \eqref{eq:coarseEst_beta}, and update $\dot{\qy}_{\mathrm{v},m,\mathrm{r}}$ via \eqref{eq:yvm_r_dot}.
       \STATE Refine $\hat{\theta}_{\max}$ and $\hat{\beta}_{1,\max}$ by $R_\mathrm{L}$ times via \eqref{eq:local_refine}.
       \STATE Update $\mathcal{P}_1 = \{\mathcal{P}_1, \hat{\theta}_{\max}\}$ and $\mathcal{Q}_1 = \{\mathcal{Q}_1, \hat{\beta}_{1,\max}\}$.
       \FOR{$i = 1, \dots, R_\mathrm{G}$}
       \FOR{each $\hat{\theta}_1$ in $\mathcal{P}_1$ and $\hat{\beta}_1$ in $\mathcal{Q}_1$}
       \STATE Refine $\hat{\theta}_1$ and $\hat{\beta}_1$ by $R_\mathrm{L}$ times via \eqref{eq:local_refine}.
       \ENDFOR
       \ENDFOR
       \STATE Update $\hat{\boldsymbol{\beta}}_1$ via \eqref{eq:beta_hat} and $\dot{\qy}_{\mathrm{v},m,\mathrm{r}} = \bar{\qy}_{\mathrm{v},m} - \qA_1 \hat{\boldsymbol{\beta}}_1$.
       \ENDWHILE
\end{algorithmic}
\end{algorithm}

\subsection{Estimation for the Azimuth AoAs and the Ranges}\label{subs:Azi_Est}
For notational convenience, the received signal at the middle row of the BS antennas is extracted as
\begin{align}
       \qy_3 = [\qy_1]_{(\frac{M_1-1}{2}M_2+1):\frac{M_1+1}{2}M_2},
\end{align}
which is equivalent to a signal vector received by a horizontal ULA, i.e., $\qy_3 = \qh_3 + \qn_3$, where $\qn_3$ denotes the received complex AWGN with $\qn_3 \sim \mathcal{C}\mathcal{N}(\mathbf{0}_{M_2 \times 1}, \sigma^2_\mathrm{n}\qI_{M_2})$, and $\qh_3 = \sqrt{\frac{M_2}{K}} \sum_{k = 1}^{K} \beta_k \tilde{\qb}_2(\vartheta_{\mathrm{h},k}, \varphi_{\mathrm{h},k})$ is the channel vector.

Due to the sparsity of the hybrid-field channel in the fractional Fourier domain, different propagation paths manifest as distinct peaks in the DFrFT spectrum.
Consequently, the estimation of the azimuth AoAs and ranges can be formulated as a sparse peak search problem over the DFrFT spectrum, without requiring prior classification of the detected paths as far-field or near-field \cite{XYang-24WCL}.
By finding the peak with the maximum amplitude in the DFrFT spectrum of the residual vector associated with $\qy_3$, which is initialized as $\qy_{3,\mathrm{r}} \triangleq [\qy_3^T, \mathbf{0}_{M_2(\nu_2-1) \times 1}^T]^T$, the coarse estimates of $\vartheta_{\mathrm{h}}$ and $\varphi_{\mathrm{h}}$ can be obtained as
\begin{align}\label{eq:coarseEst}
       \hat{\vartheta}_{\mathrm{h},\max} = \frac{2\pi \bar{j}\csc(\frac{\pi}{2}q_{\bar{i}})}{M_2\nu_2},\,
       \hat{\varphi}_{\mathrm{h},\max} = \frac{\pi \cot(\frac{\pi}{2}q_{\bar{i}})}{M_2\nu_2},
\end{align}
with $(\bar{i},\bar{j}) = \arg \mathop{\max}_{i,j} \{ |[\qG]_{i,j}|\,| \qG \triangleq [ \mathcal{F}^{(q_1)}\{\qy_{3,\mathrm{r}}\},\mathcal{F}^{(q_2)}$ $\{\qy_{3,\mathrm{r}}\}, \dots,\mathcal{F}^{(q_Q)}\{\qy_{3,\mathrm{r}}\}] \}$,
where $\nu_2$ represents the oversampling factor, $q_i = \frac{1}{2} + \frac{i-1}{Q-1}, \forall i \in \{1,2,\dots,Q\}$, $Q$ is the number of $q_i$, and we set $Q = M_2$ in the simulation.
This definition of $q_i$ guarantees a constant resolution in the fractional Fourier domain.
Accordingly, the estimate of the complex gain is given by
\begin{align}\label{eq:coarseEst_beta3}
       \hat{\beta}_{3,\max} = \frac{\{\tilde{\qb}_2(\hat{\vartheta}_{\mathrm{h},\max}, \hat{\varphi}_{\mathrm{h},\max})\}^H\tilde{\qy}_{3,\mathrm{r}}}{\|\tilde{\qb}_2(\hat{\vartheta}_{\mathrm{h},\max}, \hat{\varphi}_{\mathrm{h},\max})\|^2},
\end{align}
where $\tilde{\qy}_{3,\mathrm{r}}$ consists of the first $M_2$ elements of $\qy_{3,\mathrm{r}}$.
Then, the residual vector is updated as
\begin{align}\label{eq:y3r_dot}
       \dot{\qy}_{3,\mathrm{r}} = \bar{\qy}_{3,\mathrm{r}} - \hat{\beta}_{3,\max} \dot{\qb}_2(\hat{\vartheta}_{\mathrm{h},\max}, \hat{\varphi}_{\mathrm{h},\max}),
\end{align}
where define $\dot{\qb}_2(\hat{\vartheta}_{\mathrm{h},\max}, \hat{\varphi}_{\mathrm{h},\max}) \triangleq [\{\tilde{\qb}_2(\hat{\vartheta}_{\mathrm{h},\max}, \hat{\varphi}_{\mathrm{h},\max})\}^T,$ $\mathbf{0}_{M_2(\nu_2-1) \times 1}^T]^T$.

To alleviate the off-the-grid effect, we propose to apply the Newton-based refinements for $\hat{\vartheta}_{\mathrm{h},\max}$, $\hat{\varphi}_{\mathrm{h},\max}$ and $\hat{\beta}_{3,\max}$.
Specifically, $R_\mathrm{L}$ local refinements are implemented with the objective of minimizing the residual energy $E_3(\vartheta_{\mathrm{h}}, \varphi_{\mathrm{h}}) = \|\dot{\qy}_{3,\mathrm{r}} - \hat{\beta}_{3,\max} \dot{\qb}_2(\vartheta_{\mathrm{h}}, \varphi_{\mathrm{h}})\|^2$.
For the $(t+1)$-th iteration, we have
\begin{subequations}\label{eq:local_refine3}
\begin{align}
       &\hat{\beta}_{3,\max}^{(t+1)} = \frac{\{\dot{\qb}_2(\hat{\vartheta}_{\mathrm{h},\max}^{(t)}, \hat{\varphi}_{\mathrm{h},\max}^{(t)})\}^H}{\|\dot{\qb}_2(\hat{\vartheta}_{\mathrm{h},\max}^{(t)}, \hat{\varphi}_{\mathrm{h},\max}^{(t)})\|^2} \bigg\{ \dot{\qy}_{3,\mathrm{r}}^{(t)} \notag\\
       & \hspace{4em} + \hat{\beta}_{3,\max}^{(t)} \dot{\qb}_2(\hat{\vartheta}_{\mathrm{h},\max}^{(t)}, \hat{\varphi}_{\mathrm{h},\max}^{(t)}) \bigg\},\\
       &\dot{\qy}_{3,\mathrm{r}}^{(t+1)} = \dot{\qy}_{3,\mathrm{r}}^{(t)} - (\hat{\beta}_{3,\max}^{(t+1)} - \hat{\beta}_{3,\max}^{(t)}) \dot{\qb}_2(\hat{\vartheta}_{\mathrm{h},\max}^{(t)}, \hat{\varphi}_{\mathrm{h},\max}^{(t)}),\\
       &\left[\hat{\vartheta}_{\mathrm{h},\max}^{(t+1)}, \hat{\varphi}_{\mathrm{h},\max}^{(t+1)}\right]^T = \left[\hat{\vartheta}_{\mathrm{h},\max}^{(t)}, \hat{\varphi}_{\mathrm{h},\max}^{(t)}\right]^T - \qQ^{-1}\qx,
\end{align}
\end{subequations}
where $\qQ \in \bbC^{2 \times 2}$ denotes the Hessian matrix for $E_3(\vartheta_{\mathrm{h}}, \varphi_{\mathrm{h}})$ and $\qx \in \bbC^{2 \times 1}$ denotes the gradient vector.
Note that the updates of $\hat{\vartheta}_{\mathrm{h},\max}^{(t+1)}$ and $\hat{\varphi}_{\mathrm{h},\max}^{(t+1)}$ are adopted only when $\qQ$ is positive definite and $\|\dot{\qy}_{3,\mathrm{r}}^{(t+1)}\|^2 \leq \|\dot{\qy}_{3,\mathrm{r}}^{(t)}\|^2$, which ensures the effective reduction of residual energy and guarantees the convergence of the refinement procedure.

After local refinement, we update $\mathcal{P}_3 = \{\mathcal{P}_3, (\hat{\vartheta}_{\mathrm{h},\max}^{(t+1)},$ $\hat{\varphi}_{\mathrm{h},\max}^{(t+1)})\}$ and $\mathcal{Q}_3 = \{\mathcal{Q}_3, \hat{\beta}_{3,\max}^{(t+1)}\}$, and implement global refinements.
In this step, each $(\hat{\vartheta}_{\mathrm{h}}, \hat{\varphi}_{\mathrm{h}}) \in \mathcal{P}_3$ and $\hat{\beta}_3 \in \mathcal{Q}_3$ are iteratively refined by performing local refinement for $R_G$ times.
To further reduce the energy of the residual vector $\dot{\qy}_{3,\mathrm{r}}$, $\hat{\boldsymbol{\beta}}_3 \triangleq [\hat{\beta}_{3,1},\dots,\hat{\beta}_{3,|\mathcal{Q}_3|}]^T$ should be updated based on the LS method, as follows
\begin{align}\label{eq:beta3_hat}
       \hat{\boldsymbol{\beta}}_3 = (\qB^H\qB + \sigma_0^2\qI_{|\mathcal{Q}_3|})^{-1}\qB^H \qy_3,
\end{align}
where $\qB \triangleq [\tilde{\qb}_2(\vartheta_{\mathrm{h},1}, \varphi_{\mathrm{h},1}), \dots, \tilde{\qb}_2(\vartheta_{\mathrm{h},|\mathcal{Q}_3|}, \varphi_{\mathrm{h},|\mathcal{Q}_3|})]$.
Then, we update $\dot{\qy}_{3,\mathrm{r}} = \qy_3 - \qB \hat{\boldsymbol{\beta}}_3$.

The algorithm terminates when the residual energy falls below the expected noise energy, i.e., $\|\dot{\qy}_{3,\mathrm{r}}\|^2 < \tau_3$, where $\tau_3 = \sigma^2_\mathrm{n} M_2$ denotes the threshold.
Accordingly, the proposed Hybrid-Field DFrFT-NOMP algorithm is summarized in \textbf{Algorithm$\,$\ref{alg:DFrFT_NOMP}}.

\begin{algorithm}[!ht]
       \caption{Hybrid-Field DFrFT-NOMP Algorithm}
       \begin{algorithmic} [1]\label{alg:DFrFT_NOMP}
       \STATE \textbf{Initialize:} $\qy_{3,\mathrm{r}} \triangleq [\qy_3^T, \mathbf{0}_{M_2(\nu_2-1) \times 1}^T]^T$, and $\dot{\qy}_{3,\mathrm{r}} = \qy_{3,\mathrm{r}}$.
       \WHILE{$\|\dot{\qy}_{3,\mathrm{r}}\|^2 \geq \tau_3$}
       \STATE Obtain $(\hat{\vartheta}_{\mathrm{h},\max}, \hat{\varphi}_{\mathrm{h},\max})$ and $\hat{\beta}_{3,\max}$ based on \eqref{eq:coarseEst} and \eqref{eq:coarseEst_beta3}, and update $\dot{\qy}_{3,\mathrm{r}}$ via \eqref{eq:y3r_dot}.
       \STATE Refine $\hat{\vartheta}_{\mathrm{h},\max}$, $\hat{\varphi}_{\mathrm{h},\max}$, and $\hat{\beta}_{3,\max}$ by $R_\mathrm{L}$ times via \eqref{eq:local_refine3}.
       \STATE Update $\mathcal{P}_3 = \{\mathcal{P}_3, (\hat{\vartheta}_{\mathrm{h},\max}^{(t+1)}, \hat{\varphi}_{\mathrm{h},\max}^{(t+1)})\}$, and $\mathcal{Q}_3 = \{\mathcal{Q}_3, \hat{\beta}_{3,\max}^{(t+1)}\}$.
       \FOR{$i = 1, \dots, R_\mathrm{G}$}
       \FOR{each $(\hat{\vartheta}_{\mathrm{h}}, \hat{\varphi}_{\mathrm{h}}) \in \mathcal{P}_3$ and $\hat{\beta}_3 \in \mathcal{Q}_3$}
       \STATE Refine $\hat{\vartheta}_{\mathrm{h}}$, $\hat{\varphi}_{\mathrm{h}}$, and $\hat{\beta}_3$ by $R_\mathrm{L}$ times via \eqref{eq:local_refine3}.
       \ENDFOR
       \ENDFOR
       \STATE Update $\hat{\boldsymbol{\beta}}_3$ via \eqref{eq:beta3_hat} and $\dot{\qy}_{3,\mathrm{r}} = \qy_3 - \qB \hat{\boldsymbol{\beta}}_3$.
       \ENDWHILE
\end{algorithmic}
\end{algorithm}

\subsection{Path Matching}
In this subsection, we propose a subspace fitting-driven path matching algorithm to associate the path parameters by minimizing the subspace fitting residual.\footnote{The channel estimates of the low-dimensional subarray serve as a benchmark, as they can be easily obtained via the linear minimum mean square error estimation criterion \cite{XYang-26JIOT}.}
Given the acquired $\hat{\theta} \in \mathcal{P}_1$, $(\hat{\vartheta}_{\mathrm{h}}, \hat{\varphi}_{\mathrm{h}}) \in \mathcal{P}_3$, and $\hat{\beta} \in \mathcal{Q}_3$, the channel vector for the candidate path can be constructed as
\begin{align}\label{eq:hsub_cand}
       \qh_{\mathrm{sub}, \mathrm{cand}}(\hat{\theta},\hat{\vartheta}_{\mathrm{h}}, \hat{\varphi}_{\mathrm{h}},\hat{\beta}) = \hat{\beta} [\tilde{\qb}(\hat{\theta},\hat{\vartheta}_{\mathrm{h}}, \hat{\varphi}_{\mathrm{h}})]_{\mathcal{S}_2},
\end{align}
where $\mathcal{S}_2$ is the index set.
As a result, the path matching problem can be modeled as
\begin{align}\label{eq:path_matching_problem}
       \mathcal{G} =& \arg \mathop{\min}_{\mathcal{G} \in \mathcal{G}_0} \| \qh_{\mathrm{sub}, r} - \!\!\sum_{(\hat{\theta},\hat{\vartheta}_{\mathrm{h}}, \hat{\varphi}_{\mathrm{h}},\hat{\beta}) \in \mathcal{G}} \qh_{\mathrm{sub}, \mathrm{cand}}(\hat{\theta},\hat{\vartheta}_{\mathrm{h}}, \hat{\varphi}_{\mathrm{h}},\hat{\beta})\|_2,
\end{align}
where $\qh_{\mathrm{sub}, r}$ is the low-dimensional subarray channel estimate, $\mathcal{G}$ denotes the selected set of matched path-parameter tuples, and $\mathcal{G}_0$ contains all feasible one-to-one associations.
The proposed subspace fitting-driven path matching algorithm can be summarized in \textbf{Algorithm$\,$\ref{alg:Path_Matching}}.

\begin{algorithm}[!ht]
       \caption{Subspace Fitting-Driven Path Matching Algorithm}
       \begin{algorithmic} [1]\label{alg:Path_Matching}
       \STATE \textbf{Initialize:} $\mathcal{G} = \emptyset$.
       \FOR{each $(\hat{\vartheta}_{\mathrm{h}}, \hat{\varphi}_{\mathrm{h}}) \in \mathcal{P}_3$ and its corresponding $\hat{\beta}_3 \in \mathcal{Q}_3$}
       \FOR{each $\hat{\theta} \in \mathcal{P}_1$}
       \STATE Construct $\qh_{\mathrm{sub}, \mathrm{cand}}$ via \eqref{eq:hsub_cand}.
       \ENDFOR
       \ENDFOR
       \STATE Obtain $\mathcal{G}$ via \eqref{eq:path_matching_problem}.
\end{algorithmic}
\end{algorithm}

\subsection{Overall Channel Estimation Algorithm}
Following the above subsections, the full-dimensional channel can be recovered as
\begin{align}\label{eq:h_hat}
       \hat{\qh}_1 = \sqrt{\frac{M}{K}} \sum_{(\hat{\theta},\hat{\vartheta}_{\mathrm{h}}, \hat{\varphi}_{\mathrm{h}},\hat{\beta}) \in \mathcal{G}} \hat{\beta} \tilde{\qb}(\hat{\theta},\hat{\vartheta}_{\mathrm{h}}, \hat{\varphi}_{\mathrm{h}}),
\end{align}
and the overall low-complexity channel estimation algorithm can be summarized in \textbf{Algorithm$\,$\ref{alg:overall_estimation}}.
The convergence of \textbf{Algorithm$\,$\ref{alg:overall_estimation}} is guaranteed by the monotonic reduction of the residual energy during all Newton refinements and LS updates.
The computational complexity of \textbf{Algorithm$\,$\ref{alg:overall_estimation}} is calculated as $\mathcal{O}(K(\bar{M}_1\nu_1\log(\bar{M}_1\nu_1) + QM_2\nu_2\log(M_2\nu_2)) + K^2R_LR_G(\bar{M}_1+M_2) + K^3(\bar{M}_1+M_2) + K!K|\mathcal{S}_2|+ K^4)$.
For comparison, the computational complexity of 3D-OMP is $\mathcal{O}(K^3S+MS)$, where $S$ denotes the size of the 3D grids.
Due to the extremely large 3D joint dictionary, the prohibitive complexity of 3D-OMP renders it infeasible for practical XL-MIMO systems, which verifies the superiority of the proposed algorithm in terms of low complexity.

\begin{algorithm}[!ht]
       \caption{Overall Low-Complexity Channel Estimation Algorithm}
       \begin{algorithmic} [1]\label{alg:overall_estimation}
       \STATE \textbf{Initialize:} $\mathcal{G} = \emptyset$.
       \STATE Obtain $\mathcal{P}_1$ based on \textbf{Algorithm$\,$\ref{alg:DFT_NOMP}}.
       \STATE Obtain $\mathcal{P}_3$ and $\mathcal{Q}_3$ based on \textbf{Algorithm$\,$\ref{alg:DFrFT_NOMP}}.
       \STATE Associate the decoupled multipath parameters based on \textbf{Algorithm$\,$\ref{alg:Path_Matching}}.
       \STATE Recover $\hat{\qh}_1$ via \eqref{eq:h_hat}.
\end{algorithmic}
\end{algorithm}

\section{Correlation-Domain Aperture Extrapolation}\label{s:APEXcorr}
Although the antenna-domain extrapolation enables a computationally efficient aperture expansion, it suffers from limited elevation AoA estimation accuracy.
This limitation stems from the recursive nature of the prediction process, where the received noise and extrapolation errors accumulate rapidly, leading to severe SNR degradation.
Note that the statistical CSI, i.e., the spatial correlation matrix, can be obtained by temporally accumulating the instantaneous CSI during the uplink training period.
Interestingly, by skillfully selecting the elements from the spatial correlation matrix, a virtual array that decouples the elevation AoAs from the azimuth AoAs and ranges with enhanced SNR can be constructed, facilitating better aperture extrapolation than the antenna domain.
To this end, we establish the equivalent signal model of the constructed virtual array in the following.

\begin{remark}
       Based on the spatial correlation matrix at the BS, the equivalent received signal of the virtual array can be modeled as a standard ULA received signal, given by
       \begin{align}\label{eq:y2}
              \qy_2 = \frac{1}{K}\sum_{k=1}^{K}\sigma^2_k\tilde{\qa}(\gamma_k) + \qn_2,
       \end{align}
       where $\qn_2 \triangleq \sigma^2_\mathrm{n}[0,\dots,1,\dots,0]^T$, and $\tilde{\qa}(\gamma_k) \in \bbC^{M_1 \times 1}$ is the steering vector corresponding to the $k$-th path, defined as
       \begin{align}
              [\tilde{\qa}(\gamma_k)]_i = e^{j\frac{2\pi}{\lambda}a_i d \gamma_k},\, a_i \in \mathcal{M}_1.
       \end{align}
\end{remark}

\begin{proof}
The spatial correlation matrix of the noisy channel observation $\qh_1 + \qn_1$ can be expressed as
\begin{align}\label{eq:correlation_matrix}
       \qR =& \mathbb{E}\left\{(\qh_1 + \qn_1)(\qh_1 + \qn_1)^H\right\} \notag\\
       =& \frac{M}{K} \sum_{k=1}^{K} \sum_{p=1}^{K} \mathbb{E}\{\beta_k \beta_p^*\} \tilde{\qb}(\theta_k, \vartheta_{\mathrm{h},k}, \varphi_{\mathrm{h},k}) \notag\\
       & \hspace{5em} \times \tilde{\qb}^H(\theta_p, \vartheta_{\mathrm{h},p}, \varphi_{\mathrm{h},p}) + \sigma^2_\mathrm{n}\qI_M \notag\\
       =& \frac{M}{K} \sum_{k=1}^{K} \sigma^2_k \tilde{\qb}(\theta_k, \vartheta_{\mathrm{h},k}, \varphi_{\mathrm{h},k})  \tilde{\qb}^H(\theta_k, \vartheta_{\mathrm{h},k}, \varphi_{\mathrm{h},k}) + \sigma^2_\mathrm{n}\qI_M.
\end{align}
For analytical simplicity and in accordance with the UPA topology, $\qR$ in \eqref{eq:correlation_matrix} is partitioned into $M_1 \times M_1$ submatrices, each of size $M_2 \times M_2$, and the phase of the $(s,t)$-th element in the $(i,j)$-th submatrix can be written as \eqref{eq:R_sub}, shown at the top of the next page,
\begin{figure*}
\begin{align}\label{eq:R_sub}
       \arg [\qR(i,j)]_{(s,t)} = j\frac{2\pi}{\lambda}\sum_{k=1}^{K} \bigg\{ &(a_i - a_j) d \sin\theta_k - (a_i^2 - a_j^2) \frac{d^2\cos^2\theta_k}{2r_k} \notag\\
       & + (b_s - b_t) d\cos\theta_k\cos\phi_k -(b_s^2 - b_t^2)\frac{d^2(1-\cos^2\theta_k\cos^2\phi_k)}{2r_k} \bigg\}
\end{align}
{\noindent}\rule[-10pt]{17.5cm}{0.05em}
\end{figure*}
where the indexes $a_i,a_j \in \mathcal{M}_1$ and $b_s,b_t \in \mathcal{M}_2$.
It can be observed from \eqref{eq:R_sub} that when $a_i = -a_j$ and $b_s = b_t$ (i.e., $i+j = M_1 + 1$ and $s = t$), both the near-field quadratic phase terms and the azimuth-domain linear phase terms in the corresponding elements of the spatial correlation matrix $\qR$ are completely eliminated, yielding the expression given by
\begin{align}\label{eq:R_sub_vir}
       \arg [\qR(i,M_1 + 1-i)]_{(s,s)} =& \frac{2\pi d}{\lambda}\sum_{k=1}^{K}\left\{2a_i\sin\theta_k \right\}.
\end{align}

Based on \eqref{eq:R_sub_vir}, the elevation AoAs are decoupled from the azimuth AoAs and the ranges.
Therefore, we define the effective elevation AoA for the $k$-th path as
\begin{align}
       \gamma_k = 2\sin\theta_k.
\end{align}
By extracting and averaging the elements of $\qR$, we construct the equivalent received signal of the virtual array $\qy_2 \in \bbC^{M_1 \times 1}$, with its $i$-th element given by
\begin{align}
       [\qy_2]_i = \frac{1}{|\mathcal{M}_2|}\sum_{b_s \in \mathcal{M}_2} [\qR(i,M_1 + 1-i)]_{(s,s)}.
\end{align}
Then, $\qy_2$ can be rewritten as \eqref{eq:y2}, which completes the proof.
\end{proof}

As indicated by \eqref{eq:y2}, while correlation-domain processing achieves near-noiseless observations at non-center virtual antennas, the AWGN aggregates on the main diagonal of the spatial correlation matrix, manifesting as an impulsive noise peak at the center antenna of the virtual array.
Linear prediction-based aperture extrapolation algorithms, such as the Burg algorithm, rely on the wide-sense stationarity and the complex exponential coherent structure of the input sequence, both of which are severely disrupted by this central impulse.
As a result, the derived prediction coefficients are distorted, significantly degrading the accuracy of the aperture extrapolation.
To tackle this issue, we propose a bidirectional linear prediction-based interpolation strategy to reconstruct the corrupted observation at the center antenna of the virtual array.
Specifically, we perform forward and backward extrapolations on the sequences consisting of the first and last $\frac{M_1-1}{2}$ elements of $\qy_2$, yielding the predicted estimates $y_{\mathrm{c},\mathrm{f}}$ and $y_{\mathrm{c},\mathrm{b}}$, respectively.
Owing to the Hermitian symmetry of the spatial correlation matrix, the forward and backward predictions are conjugate, and the desired center value is real.
Hence, the observation at the center antenna of the virtual array is replaced by their average, i.e., $[\qy_2]_{\frac{M_1+1}{2}} = \frac{y_{\mathrm{c},\mathrm{f}} + y_{\mathrm{c},\mathrm{b}}}{2}$.
Despite the limited number of elements, i.e., $\frac{M_1-1}{2}$, reliable interpolation can still be achieved owing to the channel sparsity and the equivalent high-SNR condition in the correlation domain.
Subsequently, by performing antenna-domain extrapolation with linear prediction coefficients $\{\bar{c}_1,\bar{c}_2,\dots,\bar{c}_{\bar{P}}\}$ on the virtual signal $\qy_2$, we establish the following theorem.
\begin{theorem}\label{the:extra_signal_correlation_domain}
       In the correlation domain, the extrapolated signal $\bar{\qy}_2 \in \bbC^{\check{M}_1 \times 1}$ can be formulated as
       \begin{align}\label{eq:y2_bar}
              \bar{\qy}_2 = \sum_{k=1}^{K}g_k \bar{\qa}(\gamma_k) + \bar{\qn}_2,
       \end{align}
       where $g_k \triangleq \frac{1}{K}\sigma^2_k$, $\check{M}_1 \triangleq 2\lfloor \xi_2 M_1 \rfloor+M_1$, $\xi_2$ denotes the extrapolation factor, i.e., $\lfloor \xi_2 M_1 \rfloor$ represents the length of the forward/backward extrapolated sequence.
       $\bar{\qn}_2 \triangleq [\tilde{\qn}_2^T,\qn_2^T,\breve{\qn}_2^T]^T$ denotes the virtual received noise for the extrapolated array, $\breve{\qn}_2 \sim \mathcal{C}\mathcal{N}(\mathbf{0}_{\lfloor \xi_2 M_1 \rfloor \times 1}, \breve{\qR}_2)$ and $\tilde{\qn}_2 \sim \mathcal{C}\mathcal{N}(\mathbf{0}_{\lfloor \xi_2 M_1 \rfloor \times 1}, \tilde{\qR}_2)$ represent the noise induced by the forward and backward extrapolation processes with $\breve{\qR}_2$ defined in \eqref{eq:R2_breve} and $\tilde{\qR}_2$ defined in \eqref{eq:R2_tilde}, respectively.
       $\bar{\qa}(\gamma_k) \in \bbC^{\check{M}_1 \times 1}$ denotes the steering vector corresponding to the $k$-th path for the extrapolated array, as follows
       \begin{align}
              \bar{\qa}(\gamma_k) = \left[e^{j\frac{2\pi}{\lambda}(-\frac{\check{M}_1-1}{2})d\gamma_k},\dots,e^{j\frac{2\pi}{\lambda}\frac{\check{M}_1-1}{2}d\gamma_k}\right]^T.
       \end{align}
\end{theorem}

\begin{proof}
See Appendix\,\ref{proof:the:extra_signal_correlation_domain}.
\end{proof}

Then, we establish the theoretical performance limit for the correlation-domain extrapolated signal by deriving its corresponding CRB in the following theorem.

\begin{theorem}\label{the:CRBcorr}
       Based on the extrapolated signal $\bar{\qy}_2$ in \eqref{eq:y2_bar}, the FIM for the elevation AoA vector $\boldsymbol{\theta}$, with the path-power vector $\mathbf{g} \triangleq [g_1,\dots,g_K]^T$ treated as the nuisance parameter, is given by
       \begin{align}\label{eq:F_eq_corr}
              \mathbf{F}_{\mathrm{eq,corr}} =& 2 \boldsymbol{\Lambda}_{g} \bigg(\Re\{(\dot{\mathbf{A}}_{\gamma}\mathbf{J})^H \mathbf{R}_{\bar{n},2}^{-1} \dot{\mathbf{A}}_{\gamma}\mathbf{J}\} \notag\\
              & - \Re\{(\dot{\mathbf{A}}_{\gamma}\mathbf{J})^H\mathbf{R}_{\bar{n},2}^{-1} \bar{\mathbf{A}}_{\gamma}\} \{\Re\{\bar{\mathbf{A}}_{\gamma}^H \mathbf{R}_{\bar{n},2}^{-1}\bar{\mathbf{A}}_{\gamma}\}\}^{-1} \notag\\
              & \times \Re\{\bar{\mathbf{A}}_{\gamma}^H \mathbf{R}_{\bar{n},2}^{-1} \dot{\mathbf{A}}_{\gamma}\mathbf{J}\} \bigg) \boldsymbol{\Lambda}_{g},
       \end{align}
       where $\boldsymbol{\Lambda}_{g} \triangleq \mathrm{diag}(\mathbf{g})$, $\dot{\mathbf{A}}_{\gamma} \triangleq \frac{\partial \bar{\mathbf{A}}_{\gamma}}{\partial \boldsymbol{\gamma}^T}$ with $\bar{\mathbf{A}}_{\gamma}$ denoting the extrapolated array manifold matrix whose $k$-th column is $\bar{\mathbf{a}}(\gamma_k)$, $\boldsymbol{\gamma} \triangleq [\gamma_1,\dots, \gamma_K]^T$, $\mathbf{J} \triangleq \mathrm{diag}(2\cos\theta_1,\dots,2\cos\theta_K)$ is the Jacobian matrix of $\boldsymbol{\gamma}$ with respect to $\boldsymbol{\theta}$,
       and $\mathbf{R}_{\bar{n},2}$ is the correlation matrix for $\bar{\qn}_2$, defined as \eqref{eq:R_n_corr}.
       Thus, the CRB for $\theta_k$ is
       \begin{align}
              \mathrm{CRB}_{\mathrm{corr}}(\theta_k) = [\mathbf{F}_{\mathrm{eq,corr}}^{-1}]_{k,k}.
       \end{align}
\end{theorem}

\begin{proof}
See Appendix\,\ref{proof:the:CRBcorr}.
\end{proof}

The estimation of the effective elevation AoAs from $\bar{\qy}_2$ follows an identical algorithmic framework to the Ex-DFT-NOMP algorithm.
By applying the procedure in \textbf{Algorithm\,\ref{alg:DFT_NOMP}} to the signal $\bar{\qy}_2$, we obtain the parameter set $\mathcal{P}_2$, where $\hat{\gamma} \in \mathcal{P}_2$ denotes the estimated effective elevation AoAs.
Note that the termination condition is modified to $\|\dot{\qy}_{2,\mathrm{r}}\|^2 < \tau_2$, with the threshold $\tau_2$ defined as
\begin{align}
       \tau_2 = \mathrm{tr}(\breve{\qR}_2) + \mathrm{tr}(\tilde{\qR}_2) + \sigma_c^2.
\end{align}
Additionally, the path matching procedure in \textbf{Algorithm$\,$\ref{alg:Path_Matching}} is readily applicable to the mixed AoAs $\hat{\gamma} \in \mathcal{P}_2$ with $\hat{\theta} = \arcsin(\frac{\hat{\gamma}}{2})$.
Thus, the full-dimensional channel can be recovered via \textbf{Algorithm$\,$\ref{alg:overall_estimation}} by incorporating the aforementioned adjustments.

\begin{remark}
Since the spatial correlation matrix effectively suppresses ambient noise, this correlation-domain scheme inherently operates in an equivalent high-SNR regime.
Under these conditions, the accumulated prediction error rapidly overwhelms the resolution gain during the recursive process, thereby dictating a stricter upper bound on the aperture expansion than that in the antenna domain.
As evidenced by Fig.\,\ref{fig:NMSE-CRBaoa_SNR}(b), this mechanism restricts the maximum achievable extrapolated aperture to approximately $1.5$ times the original physical size.
\end{remark}

\begin{remark}
Both proposed aperture extrapolation schemes improve elevation AoA estimation accuracy and serve complementary roles within the proposed framework.
The antenna-domain extrapolation rapidly synthesizes an enlarged aperture utilizing the instantaneous received signal, making it suitable for fast-varying scenarios.
In contrast, the correlation-domain extrapolation exploits the spatial correlation matrix to mitigate ambient noise, thereby achieving higher estimation accuracy.
Therefore, the two schemes offer a flexible trade-off between rapid adaptation and estimation accuracy under varying channel conditions.
\end{remark}

\section{Numerical Results}\label{s:Simulations}
In this section, we conduct simulations to validate the effectiveness and superiority of the proposed low-complexity channel estimation framework.
In the simulations, the considered XL-MIMO system operates at $28$\,GHz, and we perform $5000$ channel realizations for each SNR.
Unless otherwise specified, we set $M_1=15$, $M_2=155$, $\bar{M}_1 = 25$, $N = 200$, $|\mathcal{S}_2| = 32$, $d=\frac{\lambda}{2}$, $\nu_1 = \nu_2 = 4$, $K_\mathrm{n} = K_\mathrm{f} =2$, $R_\mathrm{L} = R_\mathrm{G} = 5$, $\sigma_0^2 = 10^{-5}$, $\theta_k \sim \mathcal{U}(0,\frac{\pi}{4})$, $\phi_k \sim \mathcal{U}(0,\frac{\pi}{3})$, and $r_{k} \sim \mathcal{U}(3.4\,\mathrm{m}, 86.6\,\mathrm{m})$.

\begin{figure}[htbp]
       \centering
       \subfloat[]{
       \includegraphics[width=0.23\textwidth]{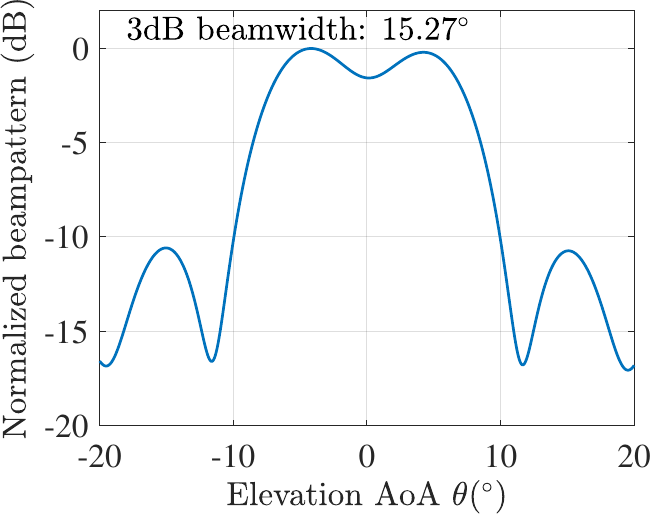}}\,
       \subfloat[]{
       \includegraphics[width=0.23\textwidth]{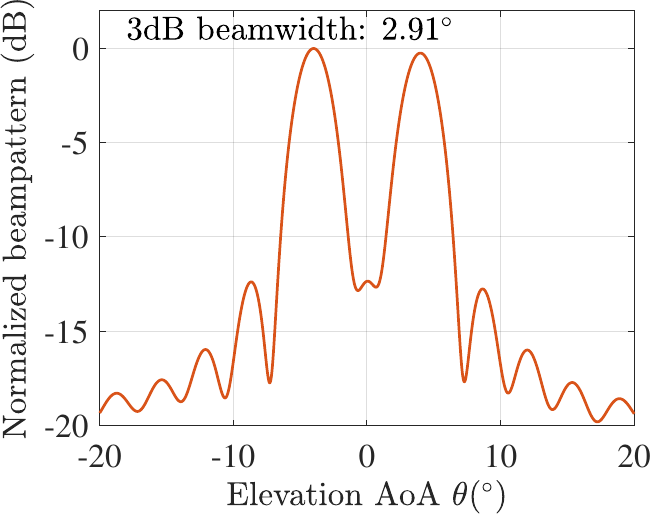}}\\
       \subfloat[]{
       \includegraphics[width=0.23\textwidth]{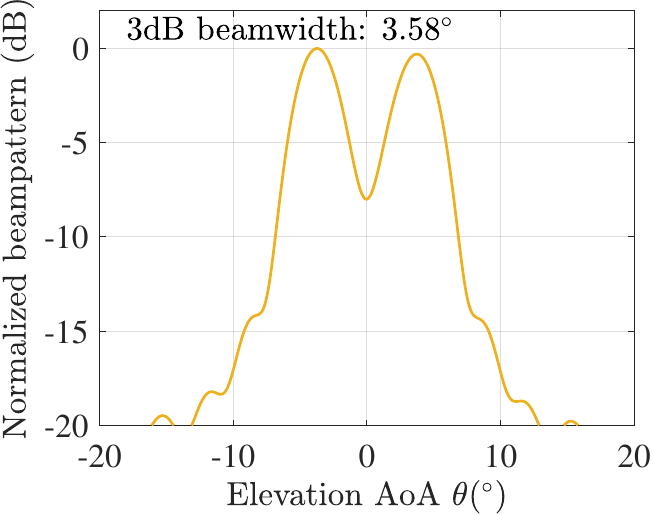}}\,
       \subfloat[]{
       \includegraphics[width=0.23\textwidth]{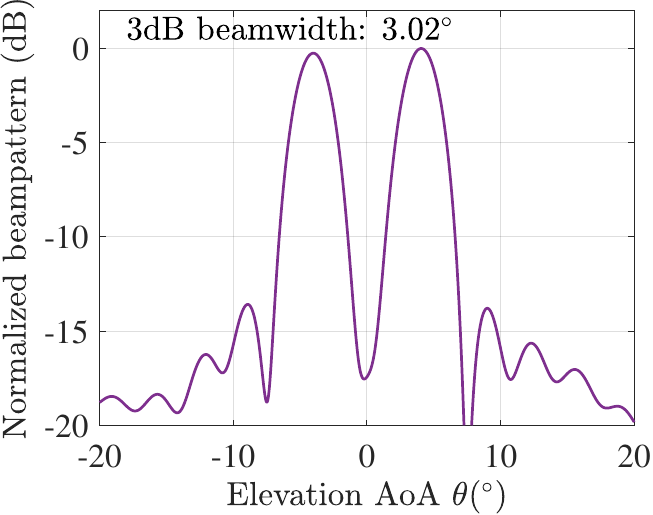}}\\
       \caption{Normalized beampattern for (a) the small aperture with $M_1 = 15$, (b) the large aperture with $M_1 = 35$, (c) the antenna-domain extrapolated aperture with $\bar{M}_1 = 35$, and (d) the correlation-domain extrapolated aperture with $\check{M}_1 = 35$.}
       \label{fig:Beampattern}
\end{figure}

\begin{figure*}[htbp]
       \centering
       \subfloat[]{
       \includegraphics[width=0.3\textwidth]{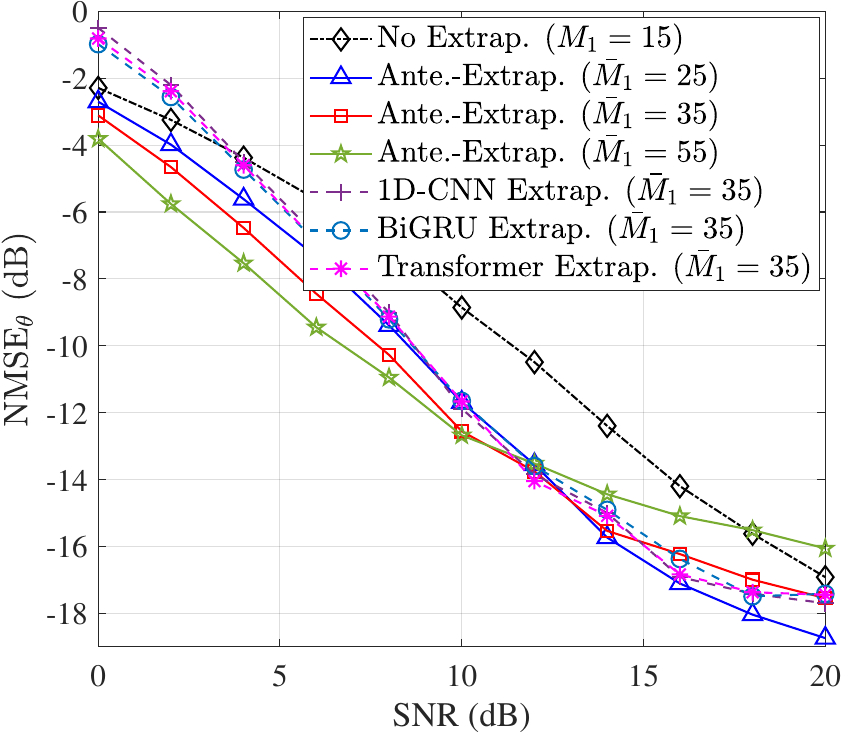}}\;
       \subfloat[]{
       \includegraphics[width=0.3\textwidth]{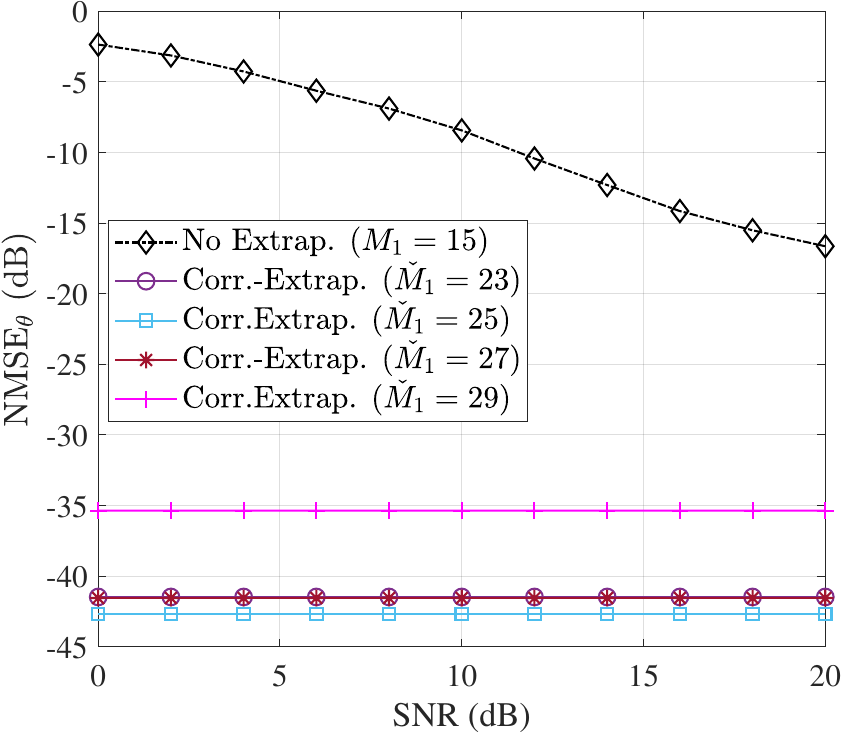}}\;
       \subfloat[]{
       \includegraphics[width=0.3\textwidth]{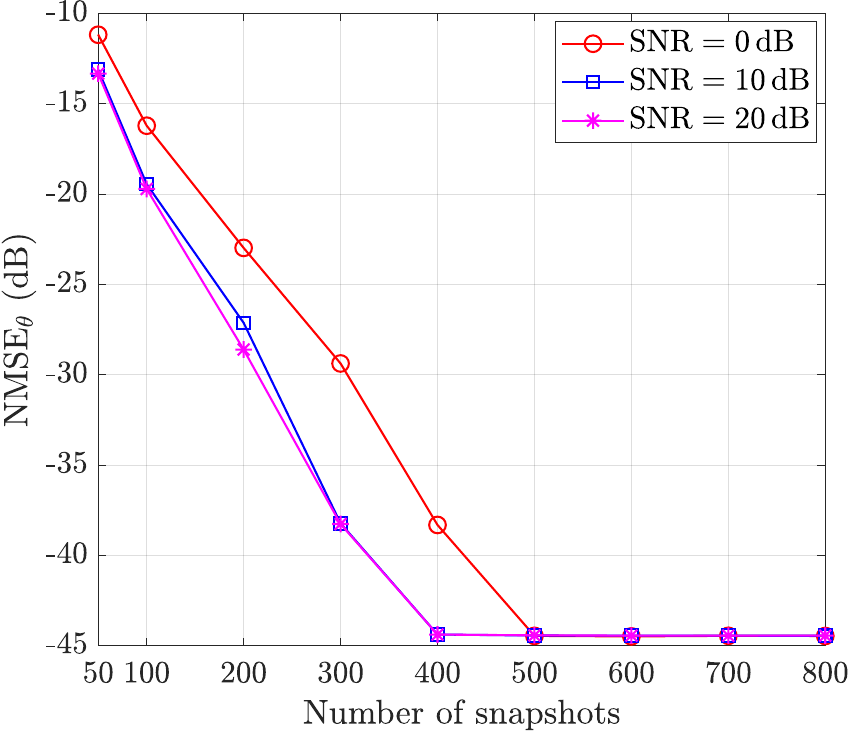}}\\
       \caption{(a) NMSE of the elevation AoA estimates versus SNR for the original array and the extrapolated arrays in the antenna domain with $\bar{M}_1 \in \{25, 35, 55\}$, (b) NMSE of the elevation AoA estimates versus SNR for the original array and the extrapolated arrays in the correlation domain with $\check{M}_1 \in \{23, 25, 27, 29\}$, and (c) NMSE of the elevation AoA estimates versus the number of snapshots at different SNRs.}\label{fig:NMSEaoa_SNR}
\end{figure*}

First, we evaluate the normalized beampatterns of different array configurations in Fig.\,\ref{fig:Beampattern} to verify the elevation AoA resolution enhancement achieved by the proposed extrapolation schemes, in which two closely spaced signals with elevation AoAs of $-4^{\circ}$ and $4^{\circ}$ are considered.
As shown in Fig.\,\ref{fig:Beampattern}(a), the small-aperture array with $M_1=15$ exhibits severe spectral peak overlapping restricted by the Rayleigh resolution limit, thereby forming a broad mainlobe with a 3-dB beamwidth of $15.27^{\circ}$ and failing to resolve the two signals.
Conversely, the physically large array with $M_1=35$ in Fig.\,\ref{fig:Beampattern}(b) resolves the two peaks, reducing the beamwidth to $2.91^{\circ}$.
Fig.\,\ref{fig:Beampattern}(c) and Fig.\,\ref{fig:Beampattern}(d) illustrate the beampatterns of the extrapolated arrays with $\bar{M}_1=35$ and $\check{M}_1=35$ generated by the antenna-domain and correlation-domain extrapolation schemes, respectively.
By synthesizing enlarged virtual vertical apertures, both extrapolation schemes narrow the spectral mainlobes and separate the previously overlapping peaks, thereby enhancing the elevation AoA resolution.
Benefiting from the noise-suppression capability of the spatial correlation matrix, the correlation-domain extrapolation exhibits sharper mainlobes (3-dB beamwidth of $3.02^{\circ}$ versus $3.58^{\circ}$) and forms a deeper null between the two closely spaced peaks compared to the antenna-domain extrapolation.
These results demonstrate the effectiveness of the proposed aperture extrapolation techniques in enhancing the AoA resolution under constrained hardware architectures.

Then, we study the estimation accuracy of the elevation AoAs.
Specifically, the normalized mean square error (NMSE) of the elevation AoA estimates $\hat{\boldsymbol{\theta}}$, defined as $\mathrm{NMSE}_\theta = \mathbb{E}(\|\boldsymbol{\theta}-\hat{\boldsymbol{\theta}}\|^2/\|\boldsymbol{\theta}\|^2)$, versus SNR for the original array and the extrapolated arrays in the antenna domain with $\bar{M}_1 \in \{25, 35, 55\}$ and in the correlation domain with $\check{M}_1 \in \{23, 25, 27, 29\}$ is depicted in Fig.\,\ref{fig:NMSEaoa_SNR}(a) and Fig.\,\ref{fig:NMSEaoa_SNR}(b), respectively.
For a comprehensive comparison, three representative learning-based baselines (1D-CNN \cite{1D-CNN}, BiGRU \cite{RNN}, and Transformer \cite{transformer}) are introduced in Fig.\,\ref{fig:NMSEaoa_SNR}(a).
These models are trained on $50000$ independent realizations to directly predict the extrapolated samples from identical noisy observations.
Overall, both the antenna-domain and correlation-domain extrapolation schemes alleviate the elevation-resolution limitation imposed by the physical aperture, and their estimation accuracy outperforms that of the original array.
Furthermore, benefiting from the noise-suppression capability of the spatial correlation matrix and the bidirectional linear prediction-based interpolation strategy, the correlation-domain extrapolation scheme achieves superior estimation accuracy compared to the antenna-domain extrapolation scheme.
As observed in Fig.\,\ref{fig:NMSEaoa_SNR}(a), for the same virtual aperture size of $\bar{M}_1 = 35$, the proposed Burg-based extrapolation outperforms the artificial intelligence (AI) baselines in the low-to-moderate SNR regime, while their performance becomes comparable at high SNRs.
This is because the noise-free received signal sequence along each vertical subarray satisfies a constant-coefficient linear difference equation, rendering heavily parameterized non-linear networks unnecessary.
Instead, this non-linear redundancy causes the networks to overfit the ambient noise, particularly in the low SNR regime.
The resulting phase distortions overwhelm the resolution gain provided by aperture expansion, causing the AI baselines to perform worse than the no extrapolation case at low SNRs.
Furthermore, the three AI models perform similarly because their diverse non-linear capacities are equally redundant for this linear task.
For the Burg-based scheme, a larger extrapolated aperture (e.g., $\bar{M}_1=55$) achieves the optimal NMSE performance in the low SNR region by sharpening the spatial spectrum, while it performs worse than smaller extrapolated apertures (e.g., $\bar{M}_1=35$ or $25$) in the high SNR region.
This performance crossover stems from the trade-off between the spatial resolution gain and the recursive prediction error accumulation during the extrapolation process.
In the low SNR regime, the ambient noise outweighs the accumulated extrapolation error, rendering the extrapolation-induced performance loss marginal.
Thus, the resolution gain from aperture expansion significantly improves the estimation performance.
However, as the SNR increases and the ambient noise power decreases, a critical threshold is reached where the accumulated extrapolation error begins to dominate the estimation accuracy.
Consequently, in the high SNR regime, the accumulated error induced by excessive extrapolation overwhelms the resolution gain, thereby degrading the estimation performance.
For practical deployment, this crossover phenomenon suggests dynamically adapting the optimal extrapolation factor by querying an offline pre-calibrated SNR-dependent look-up table.

The correlation-domain extrapolation in Fig.\,\ref{fig:NMSEaoa_SNR}(b) exhibits a similar phenomenon, where the NMSE with $\check{M}_1=29$ is higher than those with $\check{M}_1 \in \{23, 25, 27\}$.
Notably, the NMSE curves for the correlation-domain extrapolation remain nearly flat across the considered SNR regime.
This stems from the fact that the ambient noise is effectively isolated onto the main diagonal of the spatial correlation matrix and subsequently mitigated by the bidirectional interpolation, rendering the estimation accuracy dominated by the algorithmic error floor rather than the ambient noise.
In practice, the spatial correlation matrix is estimated from a finite number of received snapshots, and the resulting covariance-estimation error perturbs the equivalent virtual-array signal and subsequently affects the Burg-based interpolation and aperture extrapolation.
To quantify this effect, Fig.\,\ref{fig:NMSEaoa_SNR}(c) plots the NMSE of the elevation AoA estimates versus the number of snapshots $T$ at SNRs of $0$, $10$, and $20\,\mathrm{dB}$.
It can be observed that the NMSE decreases as $T$ increases and then levels off at approximately $-44\,\mathrm{dB}$.
This improvement stems from decreasing the sample-covariance error, yielding a more accurate equivalent virtual-array signal for subsequent interpolation and extrapolation.
Specifically, this plateau is reached at $T = 400$ for SNRs of $10$ and $20\,\mathrm{dB}$ and at $T=500$ for an SNR of $0\,\mathrm{dB}$.
This indicates that $500$ snapshots are sufficient to achieve high elevation AoA estimation accuracy under the considered simulation settings, with little further improvement as $T$ increases.
These results quantify the impact of finite-snapshot covariance estimation and illustrate a trade-off between the number of collected snapshots and the elevation AoA estimation accuracy of the correlation-domain scheme.

\begin{figure}[htbp]
       \centering
       \subfloat[]{
       \includegraphics[width=0.235\textwidth]{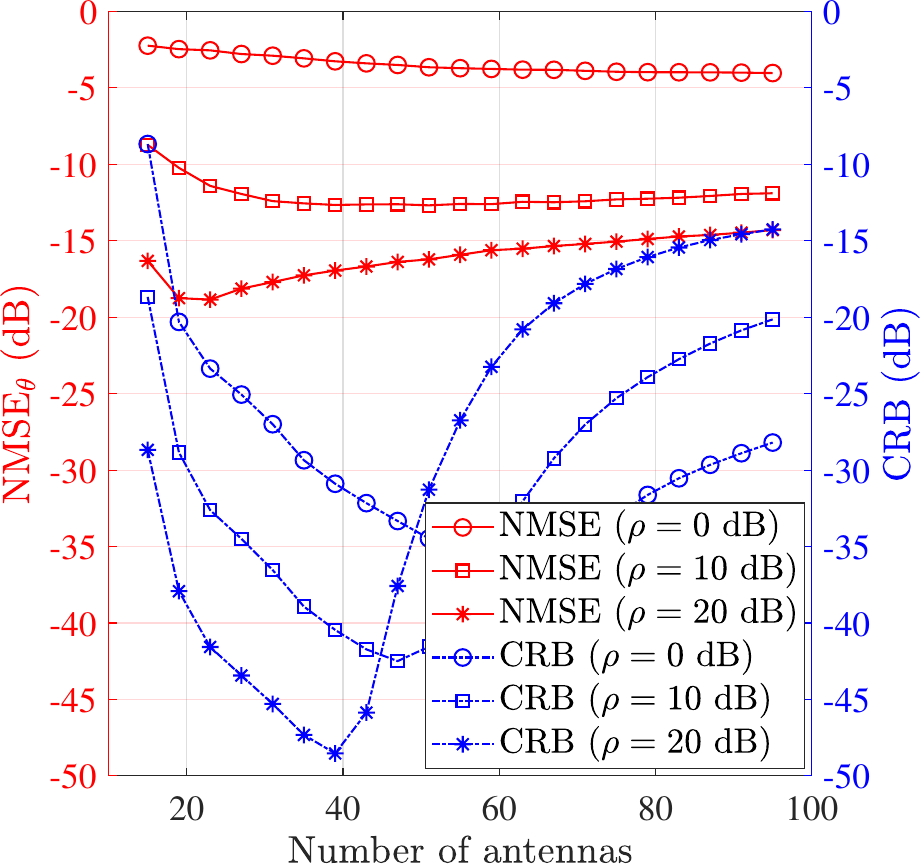}}
       \subfloat[]{
       \includegraphics[width=0.235\textwidth]{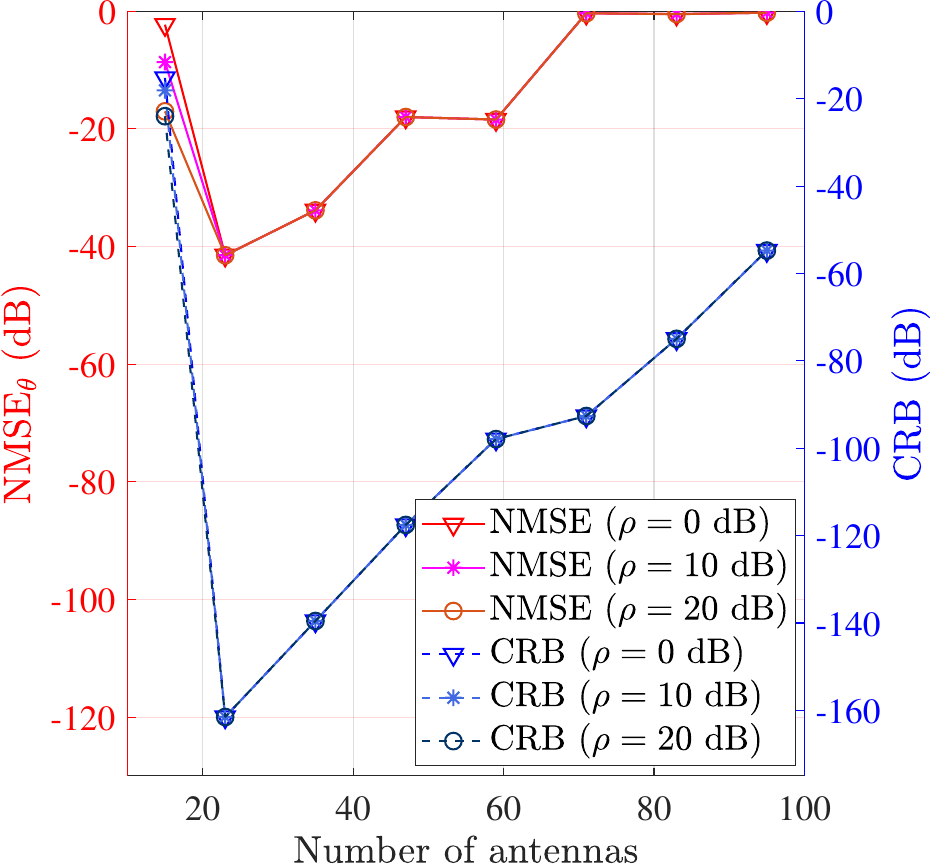}}
       \caption{NMSE and CRB of the elevation AoA estimates versus the number of antennas in the original array and the extrapolated arrays under different SNRs with $\rho \in \{0, 10, 20\}\,\mathrm{dB}$ (a) in the antenna domain, and (b) in the correlation domain.}\label{fig:NMSE-CRBaoa_SNR}
\end{figure}

To investigate the impact of the extrapolated aperture on the elevation AoA estimation performance, we depict the NMSE and CRB of the elevation AoA estimates versus the number of antennas in the original array and the extrapolated arrays under different SNRs with $\rho \in \{0, 10, 20\}\,\mathrm{dB}$ in the antenna domain and in the correlation domain in Fig.\,\ref{fig:NMSE-CRBaoa_SNR}(a) and Fig.\,\ref{fig:NMSE-CRBaoa_SNR}(b), respectively.
It can be observed that both the NMSE and CRB curves initially decrease and subsequently increase.
This phenomenon embodies a trade-off between the resolution enhancement and the prediction errors introduced and recursively amplified by aperture extrapolation.
When the number of extrapolated antennas is relatively small, the resolution gain dominates.
However, as the aperture becomes excessively large, the recursively accumulated prediction errors dictate the performance.
Furthermore, a noticeable gap exists between the NMSE and the CRB in both antenna and correlation domains.
This discrepancy arises because the aperture extrapolation process introduces colored noise, which is not effectively mitigated by the proposed Ex-DFT-NOMP algorithm.
Note that asymptotically approaching the CRB would necessitate formulating a weighted non-linear least squares estimator coupled with a multi-dimensional joint search, triggering a prohibitive computational complexity explosion.
Particularly, the low CRB floor in Fig.\,\ref{fig:NMSE-CRBaoa_SNR}(b) stems from the fact that its derivation relies on the ideal spatial correlation matrix.
By constructing a virtual array, i.e., \eqref{eq:y2}, and employing bidirectional interpolation, the ambient noise is heavily suppressed, thereby establishing a remarkably low theoretical bound.

\begin{figure}[htbp]
       \centering
       \includegraphics[width=0.3\textwidth]{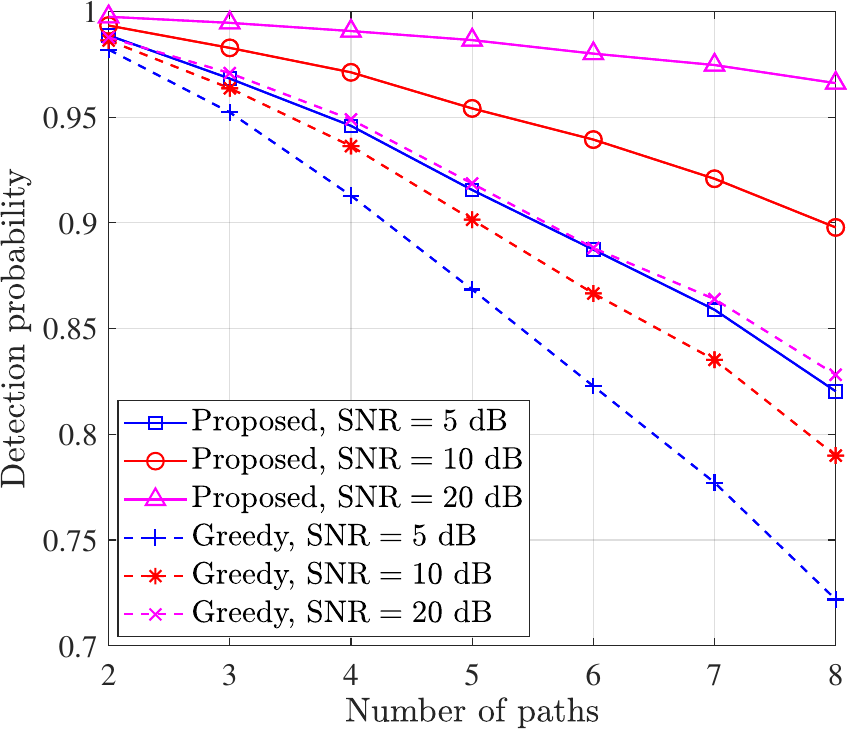}
       \caption{Detection probability of the proposed and the greedy path matching algorithms versus the number of paths under different SNRs.}\label{fig:Detection_probability}
\end{figure}

To evaluate the pairing accuracy of \textbf{Algorithm$\,$\ref{alg:Path_Matching}}, Fig.\,\ref{fig:Detection_probability} evaluates the detection probabilities of \textbf{Algorithm$\,$\ref{alg:Path_Matching}} and the greedy path matching algorithm versus the number of paths under different SNRs, where the detection probability is defined as the average proportion of correctly paired paths across Monte Carlo trials.
Compared with the greedy algorithm, the proposed \textbf{Algorithm$\,$\ref{alg:Path_Matching}} consistently achieves a higher detection probability, with a more pronounced advantage as $K$ increases.
Moreover, the detection probability decreases as $K$ increases, since a larger $K$ increases the correlation among the steering vectors and reduces the residual separation between the correct and incorrect candidate pairings.
However, an adequate SNR mitigates the effects of noise and parameter estimation errors, thereby maintaining a high detection probability.
Note that practical XL-MIMO systems typically operate at high frequencies (e.g., $28 \,\mathrm{GHz}$).
Due to severe path loss, these channels exhibit pronounced spatial sparsity, generally comprising only $3$ to $5$ dominant paths, thereby mitigating the orthogonality degradation and supporting the practical reliability of \textbf{Algorithm$\,$\ref{alg:Path_Matching}} while keeping its computational cost acceptable \cite{3GPP-38.901}, \cite{CD-22TCOMM}.

\begin{figure*}[htbp]
       \centering
       \subfloat[]{
       \includegraphics[width=0.31\textwidth]{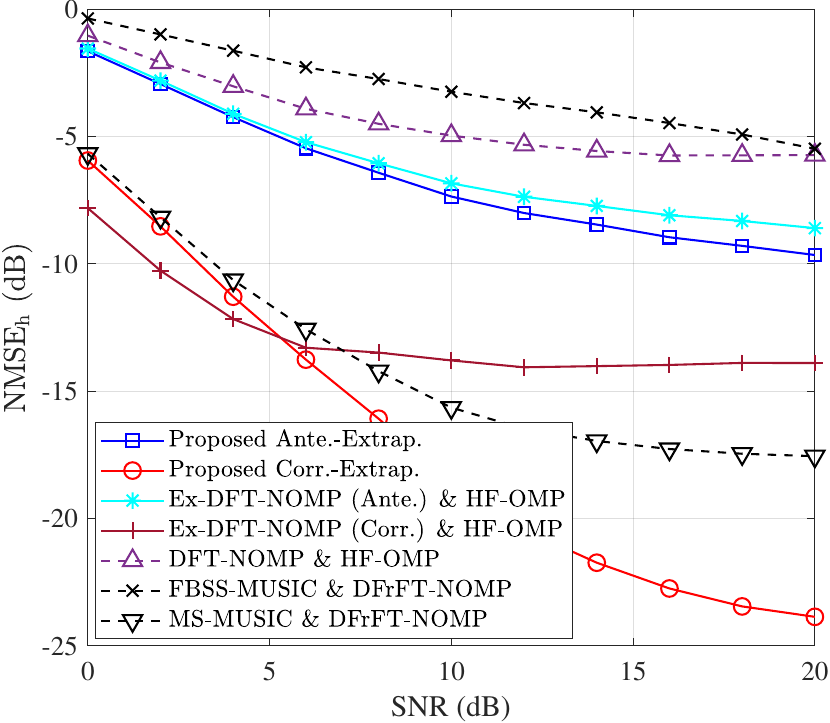}}\;
       \subfloat[]{
       \includegraphics[width=0.31\textwidth]{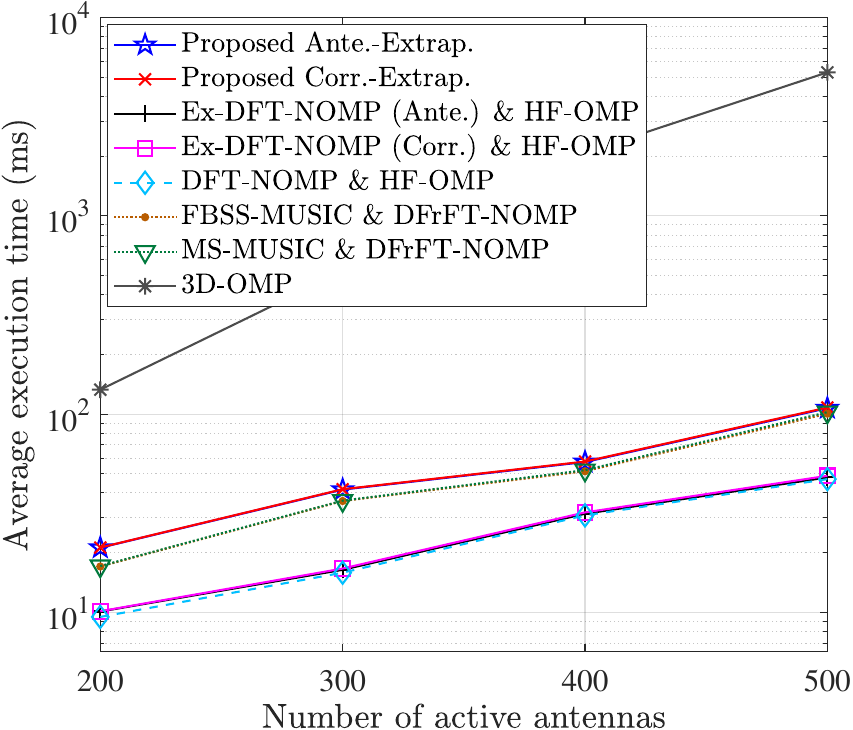}}\;
       \subfloat[]{
       \includegraphics[width=0.31\textwidth]{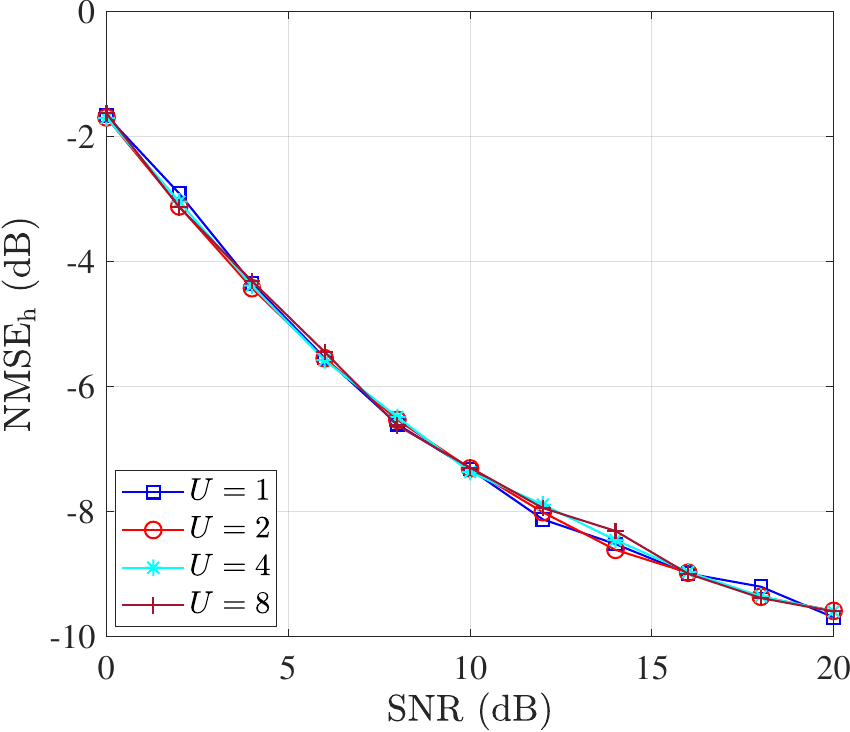}}\\
       \caption{(a) NMSE of the channel estimates versus SNR for different algorithms, (b) average execution time of the considered algorithms for different numbers of active antennas, and (c) average per-user channel NMSE versus SNR for different numbers of users.}\label{fig:NMSEh_SNR}
\end{figure*}

Fig.\,\ref{fig:NMSEh_SNR}(a) evaluates the NMSE of the channel estimates, defined as $\mathrm{NMSE}_\mathrm{h} = \mathbb{E}(\|\qh_1-\hat{\qh}_1\|^2/\|\qh_1\|^2)$, versus SNR for different algorithms.
For a fair comparison, the baselines adopt the decoupled structure and the identical path matching procedure, as direct 3D channel estimation is computationally prohibitive.
Specifically, the baselines include the forward-backward spatial smoothing (FBSS)-/multiple snapshot (MS)-MUSIC \cite{Pakrooh-16TSP} combined with DFrFT-NOMP, alongside the hybrid field (HF)-OMP \cite{WD-22CL} paired with either DFT-NOMP or Ex-DFT-NOMP.
The FBSS-MUSIC baseline exhibits the worst performance, since the limited vertical aperture renders its subspace orthogonality highly noise-sensitive.
The schemes utilizing Ex-DFT-NOMP significantly outperform the DFT-NOMP baseline, verifying that the proposed aperture extrapolation circumvents the elevation AoA resolution bottleneck.
Moreover, by executing continuous-domain Newtonized refinements, the proposed algorithm mitigates off-grid errors, thereby surpassing the OMP-based baselines.
Overall, the proposed antenna-domain scheme achieves the optimal NMSE performance among the single-snapshot algorithms, while the correlation-domain scheme exhibits the best performance against all baselines in the moderate-to-high SNR regimes, confirming the superiority of the proposed framework.

Fig.\,\ref{fig:NMSEh_SNR}(b) compares the average execution time of the considered algorithms for different numbers of active antennas.
The antenna-domain and correlation-domain implementations of the proposed framework exhibit comparable runtimes because they differ mainly in the construction of the vertical observation signal, whose cost is minor relative to the shared NOMP search and refinement procedures.
Their runtimes remain on the same order as those of the MUSIC-based baselines and increase moderately with the array size.
In contrast, direct 3D-OMP incurs rapidly increasing computational costs due to sparse recovery over a joint multidimensional dictionary.
Although the HF-OMP-based baselines are faster, they require prior knowledge of the separate numbers of near-field and far-field paths.
The proposed NOMP-based method requires neither these path counts nor a predefined near-field/far-field partition, thereby reducing its dependence on prior information in hybrid-field channel estimation.

Fig.\,\ref{fig:NMSEh_SNR}(c) plots the average per-user channel NMSE versus SNR for different numbers of users, i.e., $U \in \{1,2,4,8\}$.
As observed, the NMSE curves remain close over the entire SNR range.
This result is consistent with the use of mutually orthogonal pilots: pilot-matched filtering eliminates the channel components of the other users, such that each user has the same equivalent observation model as in the single-user case.
Thus, with orthogonal pilots and fixed per-user pilot energy, increasing the number of users does not noticeably degrade the average per-user channel estimation accuracy.

\section{Conclusion}\label{s:Conclusion}
In this paper, we proposed a low-complexity channel estimation framework for non-square UPA-assisted XL-MIMO systems.
First, an antenna-domain extrapolation scheme was developed to synthesize an enlarged vertical aperture, effectively addressing the elevation resolution bottleneck.
In addition, we derived the CRB for the elevation AoA estimates to reveal the theoretical limit of the maximum achievable extrapolated aperture.
Subsequently, an Ex-DFT-NOMP algorithm and a hybrid-field DFrFT-NOMP algorithm were provided to sequentially extract the elevation AoAs, as well as the azimuth AoAs, ranges, and complex gains.
A subspace fitting-driven path matching algorithm was then designed to pair these decoupled parameters for channel reconstruction.
To overcome the accuracy bottleneck induced by recursive noise accumulation, we further proposed a correlation-domain extrapolation scheme based on the structural properties of the spatial correlation matrix.
The CRB for the correlation-domain scheme was also derived to evaluate its noise-suppression capability.
Finally, numerical results verified the effectiveness and superiority of the proposed framework.

\appendices
\section{Proof of Lemma\,\ref{lem:extrapolation}}\label{proof:lem:extrapolation}
Let $z_k \triangleq e^{j\frac{2\pi}{\lambda}d\sin\theta_k}$, and define the equivalent amplitude as $\alpha_{k,m} \triangleq \frac{1}{\sqrt{K}}\beta_k [\tilde{\qb}_2(\vartheta_{\mathrm{h},k}, \varphi_{\mathrm{h},k})]_m e^{-j\frac{\pi(M_1+1)d\sin\theta_k}{\lambda}}$.
Then, $[\mathbf{h}_{\mathrm{v},m}]_i$ can be expressed as $[\mathbf{h}_{\mathrm{v},m}]_i=\sum_{k=1}^K\alpha_{k,m}z_k^i$.
We construct a $K$-th order annihilating polynomial with roots $\{z_k\}_{k=1}^K$ as $P(z)=\prod_{k=1}^K(1-z_kz^{-1})$.
Since $P(z_k)=0$ holds for any $k$, we have $z_k^i+\sum_{k=1}^K\bar{c}_k z_k^{i-k}=0$.
Thus, we obtain
\begin{align}
       [\qh_{\mathrm{v},m}]_i+\bar{c}_1[\qh_{\mathrm{v},m}]_{i-1}+\dots+\bar{c}_K[\qh_{\mathrm{v},m}]_{i-K}=0.
\end{align}
With distinct $\theta_k$, $\{z_k\}_{k=1}^K$ forms a set of $K$ distinct roots, and thus the coefficients $\{\bar{c}_1,\dots,\bar{c}_K\}$ are uniquely determined, which completes the proof.

\section{Proof of Theorem\,\ref{the:extra_signal_antenna_domain}}\label{proof:the:extra_signal_antenna_domain}
The extrapolated signal $\bar{\qy}_{\mathrm{v},m}$ in \eqref{eq:yvm_bar} can be straightforwardly obtained by concatenating the forward and backward extrapolated sequences.
Regarding the additive noise component $\bar{\qn}_{\mathrm{v},m} \triangleq [\tilde{\qn}_{\mathrm{v},m}^T,\qn_{\mathrm{v},m}^T,\breve{\qn}_{\mathrm{v},m}^T]^T$, the recursive extrapolation process fundamentally reshapes its statistical properties.
Since the aperture extrapolation process is essentially a linear combination, the extrapolated noise components $\tilde{\qn}_{\mathrm{v},m}$ and $\breve{\qn}_{\mathrm{v},m}$ preserve the zero-mean circularly symmetric complex Gaussian distribution.
Consequently, the subsequent derivation focuses on characterizing their variances.

By defining an augmented noise sequence $\{\breve{w}_i\}$ concatenating $\qn_{\mathrm{v},m}$ and $\breve{\qn}_{\mathrm{v},m}$, the noise induced by the forward extrapolation process can be recursively modeled as
\begin{align}\label{eq:nv_breve_m}
       [\breve{\qn}_{\mathrm{v},m}]_i = \sum_{p=1}^{P}(-c_p) \breve{w}_{M_1+i-p} + [\breve{\boldsymbol{\eta}}_1]_i,
\end{align}
where $\breve{\boldsymbol{\eta}}_1 \in \bbC^{\lfloor \xi_1 M_1 \rfloor \times 1}$ denotes the noise generated during the current linear prediction process.
According to the Wold decomposition theorem \cite{hayes1996}, $\breve{\boldsymbol{\eta}}_1$ is a zero-mean white noise sequence.
Furthermore, given the Gaussian nature of the background noise and the linear property of the Burg algorithm, $\breve{\boldsymbol{\eta}}_1$ can be modeled as a circularly symmetric complex Gaussian process \cite{schreier10}, i.e., $\breve{\boldsymbol{\eta}}_1 \sim \mathcal{C}\mathcal{N}(\mathbf{0}_{\lfloor \xi_1 M_1 \rfloor \times 1}, \breve{\sigma}^2_{\eta_1}\mathbf{I}_{\lfloor \xi_1 M_1 \rfloor})$.
The variance $\breve{\sigma}^2_{\eta_1}$ represents the power of the prediction error, which will be directly output by the Burg algorithm without specific estimation.
Rearranging \eqref{eq:nv_breve_m} into a matrix form yields
\begin{align}\label{eq:nv_breve_matrix}
       \breve{\qn}_{\mathrm{v},m} = \breve{\qC}_1^{-1}(\breve{\boldsymbol{\eta}}_1 - \breve{\qD}_1 \breve{\qn}_{\mathrm{r},m}),
\end{align}
where $\breve{\qD}_1 \triangleq [\dot{\qD}_1; \mathbf{0}_{(\lfloor \xi_1 M_1 \rfloor - P) \times P}]$, $\breve{\qC}_1 \in \bbC^{\lfloor \xi_1 M_1 \rfloor \times \lfloor \xi_1 M_1 \rfloor}$, $\dot{\qD}_1 \in \bbC^{P \times P}$, and $\breve{\qn}_{\mathrm{r},m} \in \bbC^{P \times 1}$ are respectively defined as
\begin{subequations}
\begin{align}
       \breve{\qC}_1 \triangleq &
       \begin{bmatrix}
       1 & 0 & 0 & \dots & 0 \\
       c_1 & 1 & 0 & \dots & 0 \\
       c_2 & c_1 & 1 & \dots & 0 \\
       \vdots & \vdots & \vdots & \ddots & 0 \\
       0 & \dots & c_2 & c_1 & 1
       \end{bmatrix},\\
       \dot{\qD}_1 \triangleq &
       \begin{bmatrix}
       c_P & c_{P-1} & c_{P-2} & \dots & c_1 \\
       0 & c_P & c_{P-1} & \dots & c_2 \\
       0 & 0 & c_P & \dots & c_3 \\
       \vdots & \vdots & \vdots & \ddots & \vdots \\
       0 & 0 & 0 & \dots & c_P
       \end{bmatrix},\\
       \breve{\qn}_{\mathrm{r},m} \triangleq & \left[[\qn_{\mathrm{v},m}]_{M_1+1-P}, [\qn_{\mathrm{v},m}]_{M_1+2-P}, \dots, [\qn_{\mathrm{v},m}]_{M_1}\right]^T.
\end{align}
\end{subequations}
Thus, the covariance matrix for $\breve{\qn}_{\mathrm{v},m}$ is given by
\begin{align}\label{eq:R_v_breve}
       \breve{\qR}_{\mathrm{v},m} =& \mathbb{E} \left\{\breve{\qn}_{\mathrm{v},m}\breve{\qn}_{\mathrm{v},m}^H \right\} \notag\\
       \overset{\mathrm{(a)}}{=}& \breve{\qC}_1^{-1}\left(\mathbb{E}\{\breve{\boldsymbol{\eta}}_1\breve{\boldsymbol{\eta}}_1^H\} + \breve{\qD}_1\mathbb{E}\{\breve{\qn}_{\mathrm{r},m}\breve{\qn}_{\mathrm{r},m}^H\}\breve{\qD}_1^H\right)(\breve{\qC}_1^{-1})^H \notag\\
       =& \breve{\qC}_1^{-1}\left(\breve{\sigma}^2_{\eta_1}\qI_{\lfloor \xi_1 M_1 \rfloor} + \sigma^2_\mathrm{n}\breve{\qD}_1\breve{\qD}_1^H\right)(\breve{\qC}_1^{-1})^H,
\end{align}
where the procedure (a) is that $\mathbb{E}\{\breve{\boldsymbol{\eta}}_1(\breve{\qD}_1 \breve{\qn}_{\mathrm{r},m})^H\} = \mathbb{E}\{\breve{\qD}_1 \breve{\qn}_{\mathrm{r},m}\breve{\boldsymbol{\eta}}_1^H\} = 0$ as $\breve{\boldsymbol{\eta}}_1$ and $\breve{\qD}_1 \breve{\qn}_{\mathrm{r},m}$ are mutually independent and $\mathbb{E}\{\breve{\boldsymbol{\eta}}_1\} = 0$.

Then, by defining an augmented noise sequence $\{\tilde{w}_i\}$ concatenating $\tilde{\qn}_{\mathrm{v},m}$ and $\qn_{\mathrm{v},m}$, the noise induced by the backward extrapolation process can be modeled as
\begin{align}\label{eq:nv_tilde_m}
       [\tilde{\qn}_{\mathrm{v},m}]_{i} = \sum_{p=1}^{P}(-c_p^{*}) \tilde{w}_{P+1+i-p} + [\tilde{\boldsymbol{\eta}}_1]_i,
\end{align}
where $\tilde{\boldsymbol{\eta}}_1 \in \bbC^{\lfloor \xi_1 M_1 \rfloor \times 1}$ denotes the noise generated during the backward linear prediction process satisfying $\tilde{\boldsymbol{\eta}}_1 \sim \mathcal{C}\mathcal{N}(\mathbf{0}_{\lfloor \xi_1 M_1 \rfloor \times 1}, \tilde{\sigma}^2_{\eta_1}\qI_{\lfloor \xi_1 M_1 \rfloor})$.
Similar to rearranging \eqref{eq:nv_breve_m} into \eqref{eq:nv_breve_matrix}, rearranging \eqref{eq:nv_tilde_m} into a matrix form yields
\begin{align}\label{eq:nv_tilde_matrix}
       \tilde{\qn}_{\mathrm{v},m} = (\tilde{\qC}_1^{*})^{-1}(\tilde{\boldsymbol{\eta}}_1 - \tilde{\qD}_1^{*} \tilde{\qn}_{\mathrm{r},m}),
\end{align}
where $\tilde{\qD}_1 \triangleq [\ddot{\qD}_1; \mathbf{0}_{(\lfloor \xi_1 M_1 \rfloor - P) \times P}]$, $\tilde{\qC}_1 \in \bbC^{\lfloor \xi_1 M_1 \rfloor \times \lfloor \xi_1 M_1 \rfloor}$, $\ddot{\qD}_1 \in \bbC^{P \times P}$, and $\tilde{\qn}_{\mathrm{r},m} \in \bbC^{P \times 1}$ are respectively defined as
\begin{subequations}
\begin{align}
       \tilde{\qC}_1 \triangleq &
       \begin{bmatrix}
       1 & c_P & c_{P-1} & \dots & 0 \\
       0 & \ddots & \vdots & \vdots & \vdots \\
       0 & 0 & 1 & c_P & c_{P-1} \\
       0 & 0 & 0 & 1 & c_P \\
       0 & 0 & 0 & 0 & 1
       \end{bmatrix},\\
       \ddot{\qD}_1 \triangleq &
       \begin{bmatrix}
       c_1 & 0 & 0 & \dots & 0 \\
       c_2 & c_1 & 0 & \dots & 0 \\
       c_3 & c_2 & c_1 & \dots & 0 \\
       \vdots & \vdots & \vdots & \ddots & \vdots \\
       c_P & c_{P-1} & c_{P-2} & \dots & c_1
       \end{bmatrix},\\
       \tilde{\qn}_{\mathrm{r},m} \triangleq & \left[[\qn_{\mathrm{v},m}]_1, [\qn_{\mathrm{v},m}]_2, \dots, [\qn_{\mathrm{v},m}]_P\right]^T.
\end{align}
\end{subequations}
Thus, the covariance matrix for $\tilde{\qn}_{\mathrm{v},m}$ is derived as
\begin{align}\label{eq:R_v_tilde}
       \tilde{\qR}_{\mathrm{v},m} =& \mathbb{E} \left\{\tilde{\qn}_{\mathrm{v},m}\tilde{\qn}_{\mathrm{v},m}^H \right\} \notag\\
       =& (\tilde{\qC}_1^{*})^{-1}\!\left\{\tilde{\sigma}^2_{\eta_1}\qI_{\lfloor \xi_1 M_1 \rfloor} + \sigma^2_\mathrm{n}(\tilde{\qD}_1^{*})(\tilde{\qD}_1^{*})^H\right\}\left\{(\tilde{\qC}_1^{*})^{-1}\right\}^H\!\!,
\end{align}
which completes the proof.

\section{Proof of Theorem\,\ref{the:CRBante}}\label{proof:the:CRBante}
According to {\it Theorem\,\ref{the:extra_signal_antenna_domain}}, the extrapolated observation vector in the antenna domain can be reformulated as
\begin{align}
       \bar{\qy}_{\mathrm{v},m} = \boldsymbol{\mu}_1(\boldsymbol{\xi}_1) + \bar{\qn}_{\mathrm{v},m},
\end{align}
where $\boldsymbol{\mu}_1(\boldsymbol{\xi}_1) \triangleq \bar{\mathbf{A}}\boldsymbol{\beta}$ denotes the deterministic noise-free signal component,
$\boldsymbol{\xi}_1 \triangleq [\boldsymbol{\theta}^T, \boldsymbol{\beta}_{\mathrm{R}}^T, \boldsymbol{\beta}_{\mathrm{I}}^T]^T \in \mathbb{R}^{3K \times 1}$ is the real-valued parameter vector, $\boldsymbol{\beta}_{\mathrm{R}} \triangleq \Re(\boldsymbol{\beta})$, $\boldsymbol{\beta}_{\mathrm{I}} \triangleq \Im(\boldsymbol{\beta})$,
and we have $\bar{\qn}_{\mathrm{v},m} \sim \mathcal{C}\mathcal{N}(\mathbf{0}_{\bar{M}_1 \times 1}, \mathbf{R}_{\bar{n},\mathrm{v},m})$ with $\mathbf{R}_{\bar{n},\mathrm{v},m} \triangleq \mathbb{E}[\bar{\qn}_{\mathrm{v},m}\bar{\qn}_{\mathrm{v},m}^H]$ representing
\begin{align}\label{eq:R_n_ante}
       \mathbf{R}_{\bar{n},\mathrm{v},m} =&
       \begin{bmatrix}
              \tilde{\qR}_{\mathrm{v},m} & \mathbf{R}_{\mathrm{b},1} & \mathbf{R}_{\mathrm{f,b},1} \\
              \mathbf{R}_{\mathrm{b},1}^H & \sigma^2_\mathrm{n}\qI_{M_1} & \mathbf{R}_{\mathrm{f},1} \\
              \mathbf{R}_{\mathrm{f,b},1}^H & \mathbf{R}_{\mathrm{f},1}^H & \breve{\qR}_{\mathrm{v},m}
       \end{bmatrix},
\end{align}
where $\mathbf{R}_{\mathrm{f},1}$, $\mathbf{R}_{\mathrm{b},1}$, and $\mathbf{R}_{\mathrm{f,b},1}$ are respectively defined as
\begin{subequations}
\begin{align}
       \mathbf{R}_{\mathrm{f},1} &\triangleq \mathbb{E}\{\qn_{\mathrm{v},m}\breve{\qn}_{\mathrm{v},m}^H\} \notag\\
       & = \mathbb{E}\left\{\qn_{\mathrm{v},m} \{\breve{\qC}_1^{-1}(\breve{\boldsymbol{\eta}}_1 - \breve{\qD}_1 \breve{\qn}_{\mathrm{r},m})\}^H \right\} \notag\\
       & = \sigma^2_\mathrm{n} (-\breve{\qC}_1^{-1}\breve{\qD}_1)^H,\\
       \mathbf{R}_{\mathrm{b},1} &\triangleq \mathbb{E}\{\tilde{\qn}_{\mathrm{v},m}\qn_{\mathrm{v},m}^H\} \notag\\
       & = \mathbb{E}\left\{ (\tilde{\qC}_1^{*})^{-1}(\tilde{\boldsymbol{\eta}}_1 - \tilde{\qD}_1^{*} \tilde{\qn}_{\mathrm{r},m}) \qn_{\mathrm{v},m}^H \right\} \notag\\
       & = -\sigma^2_\mathrm{n} (\tilde{\qC}_1^{*})^{-1}\tilde{\qD}_1^{*},\\
       \mathbf{R}_{\mathrm{f,b},1} &\triangleq \mathbb{E}\{\tilde{\qn}_{\mathrm{v},m}\breve{\qn}_{\mathrm{v},m}^H\} \notag\\
       & = \mathbb{E}\left\{ \!(\tilde{\qC}_1^{*})^{-1}(\tilde{\boldsymbol{\eta}}_1 - \tilde{\qD}_1^{*} \tilde{\qn}_{\mathrm{r},m}) \{\breve{\qC}_1^{-1}(\breve{\boldsymbol{\eta}}_1 - \breve{\qD}_1 \breve{\qn}_{\mathrm{r},m})\}^H \!\right\} \notag\\
       & = \sigma^2_\mathrm{n} (\tilde{\qC}_1^{*})^{-1}\tilde{\qD}_1^{*}(\breve{\qC}_1^{-1}\breve{\qD}_1)^H.
\end{align}
\end{subequations}
By utilizing the Slepian-Bangs formula \cite{Slepian-54TIT}, the $(i, j)$-th element of the FIM is given by
\begin{align}
       \mathbf{F}_{i,j} = 2\Re \left\{ \left( \frac{\partial \boldsymbol{\mu}_1}{\partial [\boldsymbol{\xi}_1]_i} \right)^H \mathbf{R}_{\bar{n},\mathrm{v},m}^{-1} \left( \frac{\partial \boldsymbol{\mu}_1}{\partial [\boldsymbol{\xi}_1]_j} \right) \right\}.
\end{align}
Then, we partition the FIM $\mathbf{F}$ into a $2 \times 2$ block structure separating the parameters of interest $\boldsymbol{\theta}$ from the nuisance parameters $\{\boldsymbol{\beta}_R, \boldsymbol{\beta}_I\}$, as follows:
\begin{align}
       \mathbf{F}(\boldsymbol{\xi}_1) = 
       \begin{bmatrix}
              \mathbf{F}_{\theta,1} & \mathbf{F}_{\theta, \beta} \\
              \mathbf{F}_{\theta, \beta}^T & \mathbf{F}_{\beta}
       \end{bmatrix},
\end{align}
where
\begin{subequations}\label{eq:Ftheta-beta}
\begin{align}
       \mathbf{F}_{\theta,1} =& 2 \Re \{(\dot{\mathbf{A}}_{\theta}\boldsymbol{\Lambda}_{\beta})^H \mathbf{R}_{\bar{n},\mathrm{v},m}^{-1}(\dot{\mathbf{A}}_{\theta}\boldsymbol{\Lambda}_{\beta})\},\\
       \mathbf{F}_{\beta} =& 2 \begin{bmatrix}
              \Re(\boldsymbol{\Omega}) & -\Im(\boldsymbol{\Omega}) \\
              \Im(\boldsymbol{\Omega}) & \Re(\boldsymbol{\Omega})
       \end{bmatrix},\\
       \mathbf{F}_{\theta, \beta} =& 2 \begin{bmatrix}
              \Re(\boldsymbol{\Upsilon}^H) & -\Im(\boldsymbol{\Upsilon}^H)
       \end{bmatrix},
\end{align}
\end{subequations}
where $\boldsymbol{\Omega} \triangleq \bar{\mathbf{A}}^H \mathbf{R}_{\bar{n},\mathrm{v},m}^{-1}\bar{\mathbf{A}}$ and $\boldsymbol{\Upsilon} \triangleq \bar{\mathbf{A}}^H \mathbf{R}_{\bar{n},\mathrm{v},m}^{-1} \dot{\mathbf{A}}_{\theta} \boldsymbol{\Lambda}_{\beta}$.

To isolate the Fisher information exclusively pertaining to the elevation AoAs $\boldsymbol{\theta}$, we apply the Schur complement to eliminate the nuisance parameters, yielding the equivalent FIM as \cite{zhang2006schur}
\begin{align}\label{eq:Feq0}
       \mathbf{F}_{\mathrm{eq,ante}} = \mathbf{F}_{\theta,1} - \mathbf{F}_{\theta, \beta} \mathbf{F}_{\beta}^{-1} \mathbf{F}_{\theta, \beta}^T.
\end{align}
By utilizing the real-composite matrix representation of complex matrices \cite{schreier10}, the inversion of the real block matrix $\mathbf{F}_{\beta}$ can be directly mapped to the complex domain as
\begin{align}\label{eq:Fbeta_inv}
       \mathbf{F}_{\beta}^{-1} = \frac{1}{2} \mathcal{I}(\boldsymbol{\Omega}^{-1}),
\end{align}
where $\mathcal{I}(\cdot)$ denotes the mapping operator such that
\begin{align}
       \mathcal{I}(\boldsymbol{\Omega}^{-1}) = \begin{bmatrix}
              \Re(\boldsymbol{\Omega}^{-1}) & -\Im(\boldsymbol{\Omega}^{-1})\\
              \Im(\boldsymbol{\Omega}^{-1}) & \Re(\boldsymbol{\Omega}^{-1})
       \end{bmatrix}.
\end{align}
By substituting \eqref{eq:Fbeta_inv} and \eqref{eq:Ftheta-beta} into \eqref{eq:Feq0}, we obtain
\begin{align}
       \mathbf{F}_{\mathrm{eq,ante}} =& \mathbf{F}_{\theta,1} - 2 \Re \left\{ \mathbf{\Upsilon}^H \boldsymbol{\Omega}^{-1} \mathbf{\Upsilon} \right\}\notag\\
       =& 2 \Re \left\{ \boldsymbol{\Lambda}_{\beta}^H \dot{\mathbf{A}}_{\theta}^H \mathbf{R}_{\bar{n},\mathrm{v},m}^{-1} \mathbf{P}_{\mathrm{A}} \dot{\mathbf{A}}_{\theta} \boldsymbol{\Lambda}_{\beta} \right\},
\end{align}
where $\mathbf{P}_{\mathrm{A}} \triangleq \mathbf{I}_{\bar{M}_1} - \bar{\mathbf{A}}(\bar{\mathbf{A}}^H \mathbf{R}_{\bar{n},\mathrm{v},m}^{-1} \bar{\mathbf{A}})^{-1} \bar{\mathbf{A}}^H \mathbf{R}_{\bar{n},\mathrm{v},m}^{-1}$, which completes the proof.

\section{Proof of Theorem\,\ref{the:extra_signal_correlation_domain}}\label{proof:the:extra_signal_correlation_domain}
Similar to Appendix \ref{proof:the:extra_signal_antenna_domain}, the extrapolated signal $\bar{\qy}_2$ in \eqref{eq:y2_bar} can be straightforwardly constructed via sequence concatenation.
Regarding the additive noise component $\bar{\qn}_2 \triangleq [\tilde{\qn}_2^T,\qn_2^T,\breve{\qn}_2^T]^T$, the bidirectional interpolation at the central antenna and the linear prediction processes introduce stochastic errors.
Consequently, the covariance matrix for $\breve{\qn}_2$ can be given by
\begin{align}\label{eq:R2_breve0}
       \breve{\qR}_2 =& \mathbb{E} \left\{\breve{\qn}_2\breve{\qn}_2^H \right\} \notag\\
       =& \mathbb{E} \left\{(\breve{\qC}_2^{-1}\breve{\boldsymbol{\eta}}_2 - \breve{\qC}_2^{-1}\breve{\qD}_2 \breve{\qn}_{\mathrm{r},2})(\breve{\qC}_2^{-1}\breve{\boldsymbol{\eta}}_2 - \breve{\qC}_2^{-1}\breve{\qD}_2 \breve{\qn}_{\mathrm{r},2})^H \right\} \notag\\
       =& \breve{\qC}_2^{-1}\left(\breve{\sigma}^2_{\eta_2}\qI_{\lfloor \xi_2 M_1 \rfloor} + \breve{\qD}_2 \mathbb{E}\{\breve{\qn}_{\mathrm{r},2}\breve{\qn}_{\mathrm{r},2}^H\}\breve{\qD}_2^H \right)(\breve{\qC}_2^{-1})^H,
\end{align}
where $\breve{\qC}_2$ and $\breve{\qD}_2$ are obtained from $\breve{\qC}_1$ and $\breve{\qD}_1$ by replacing the linear prediction coefficients $\{c_1,c_2,\dots,c_P\}$ with $\{\bar{c}_1,\bar{c}_2,\dots,\bar{c}_{\bar{P}}\}$, respectively,
$\breve{\boldsymbol{\eta}}_2 \in \bbC^{\lfloor \xi_2 M_1 \rfloor \times 1}$ denotes the noise generated during the forward linear prediction process satisfying $\breve{\boldsymbol{\eta}}_2 \sim \mathcal{C}\mathcal{N}(\mathbf{0}_{\lfloor \xi_2 M_1 \rfloor \times 1}, \breve{\sigma}^2_{\eta_2}\qI_{\lfloor \xi_2 M_1 \rfloor})$,
and $\breve{\qn}_{\mathrm{r},2} \in \bbC^{\bar{P} \times 1}$ can be approximated as
\begin{align}
       \breve{\qn}_{\mathrm{r},2} \approx
       \begin{cases}
              \mathbf{0}_{\bar{P} \times 1}, & \bar{P} < \frac{M_1+1}{2},\\
              \left[\mathbf{0}_{1 \times (\bar{P}-\frac{M_1+1}{2})}, n_c, \mathbf{0}_{1 \times \frac{M_1-1}{2}}\right]^T, & \bar{P} \geq \frac{M_1+1}{2},\\
       \end{cases}
\end{align}
where $n_c$ denotes the zero-mean Gaussian interpolation error at the central virtual antenna with the variance of $\sigma_c^2 = \frac{1}{2}\breve{\sigma}^2_{\eta_2}$.
Then, $\breve{\qR}_2$ in \eqref{eq:R2_breve0} can be rewritten as
\begin{align}\label{eq:R2_breve}
       \breve{\qR}_2 =
       \begin{cases}
       \breve{\sigma}^2_{\eta_2}\breve{\qC}_2^{-1}(\breve{\qC}_2^{-1})^H, & \hspace{-0.5em}\bar{P} < \frac{M_1+1}{2},\\
       \breve{\qC}_2^{-1}\left(\breve{\sigma}^2_{\eta_2}\qI_{\lfloor \xi_2 M_1 \rfloor} + \sigma_c^2\breve{\qd}\breve{\qd}^H \right)(\breve{\qC}_2^{-1})^H, & \hspace{-0.5em}\bar{P} \geq \frac{M_1+1}{2},
       \end{cases}
\end{align}
where $\breve{\qd} \triangleq [\breve{\qD}_2]_{:,(\bar{P}-\frac{M_1-1}{2})}$.
Subsequently, the covariance matrix $\tilde{\qR}_2 = \mathbb{E} \{\tilde{\qn}_2\tilde{\qn}_2^H \}$ can also be derived as
\begin{align}\label{eq:R2_tilde}
       \tilde{\qR}_2 =
       \begin{cases}
       \tilde{\sigma}^2_{\eta_2}\tilde{\qC}_2^{-1}(\tilde{\qC}_2^{-1})^H, & \hspace{-0.5em}\bar{P} < \frac{M_1+1}{2},\\
       \tilde{\qC}_2^{-1}\left(\tilde{\sigma}^2_{\eta_2}\qI_{\lfloor \xi_2 M_1 \rfloor} + \sigma_c^2\tilde{\qd}\tilde{\qd}^H \right) (\tilde{\qC}_2^{-1})^H, & \hspace{-0.5em}\bar{P} \geq \frac{M_1+1}{2},
       \end{cases}
\end{align}
where $\tilde{\qC}_2$ and $\tilde{\qD}_2$ are obtained from $\tilde{\qC}_1$ and $\tilde{\qD}_1$ by replacing the linear prediction coefficients $\{c_1,c_2,\dots,c_P\}$ with $\{\bar{c}_1^{*},\bar{c}_2^{*},\dots,\bar{c}_{\bar{P}}^{*}\}$, respectively,
$\tilde{\qd} \triangleq [\tilde{\qD}_2^{*}]_{:,(\bar{P}-\frac{M_1-1}{2})}$,
and $\tilde{\sigma}^2_{\eta_2}$ denotes the variance for the noise generated during the backward linear prediction process, which completes the proof.

\section{Proof of Theorem\,\ref{the:CRBcorr}}\label{proof:the:CRBcorr}
According to {\it Theorem\,\ref{the:extra_signal_correlation_domain}}, the extrapolated observation vector in the correlation domain can be formulated as
\begin{align}
       \bar{\mathbf{y}}_2 = \boldsymbol{\mu}_2(\boldsymbol{\xi}_2) + \bar{\mathbf{n}}_2,
\end{align}
where $\boldsymbol{\mu}_2(\boldsymbol{\xi}_2) \triangleq \bar{\mathbf{A}}_{\gamma}\mathbf{g}$ denotes the deterministic noise-free signal component, $\bar{\mathbf{A}}_{\gamma}$ is the array manifold matrix parameterized by the effective elevation AoAs $\boldsymbol{\gamma} \triangleq [\gamma_1,\dots,\gamma_K]^T$ with $\gamma_k = 2\sin\theta_k$, and $\boldsymbol{\xi}_2 \triangleq [\boldsymbol{\theta}^T, \mathbf{g}^T]^T \in \mathbb{R}^{2K \times 1}$ is the real-valued parameter vector, and we have $\bar{\mathbf{n}}_2 \sim \mathcal{C}\mathcal{N}(\mathbf{0}_{\check{M}_1 \times 1}, \mathbf{R}_{\bar{n},2})$ with $\mathbf{R}_{\bar{n},2} \triangleq \mathbb{E}[\bar{\qn}_2\bar{\qn}_2^H]$ denoting
\begin{align}\label{eq:R_n_corr}
       \mathbf{R}_{\bar{n},2} =&
       \begin{bmatrix}
              \tilde{\qR}_2 & \mathbf{R}_{\mathrm{b},2} & \mathbf{R}_{\mathrm{f,b},2} \\
              \mathbf{R}_{\mathrm{b},2}^H & \mathbf{R}_{\mathrm{c},2} & \mathbf{R}_{\mathrm{f},2} \\
              \mathbf{R}_{\mathrm{f,b},2}^H & \mathbf{R}_{\mathrm{f},2}^H & \breve{\qR}_2
       \end{bmatrix},
\end{align}
where $\mathbf{R}_{\mathrm{c},2} = \mathrm{diag}(0, \dots, 0, \sigma_c^2, 0, \dots, 0) \in \mathbb{C}^{M_1 \times M_1}$,
$\mathbf{R}_{\mathrm{f},2}$, $\mathbf{R}_{\mathrm{b},2}$, and $\mathbf{R}_{\mathrm{f,b},2}$ are derived following the identical steps in Appendix \ref{proof:the:CRBante}.
Then, based on the Slepian-Bangs formula, the FIM for $\boldsymbol{\xi}_2$ can be partitioned into a block structure separating the parameters of interest $\boldsymbol{\theta}$ from the nuisance parameters $\mathbf{g}$, as follows:
\begin{align}
       \mathbf{F}(\boldsymbol{\xi}_2) = \begin{bmatrix}
              \mathbf{F}_{\theta,2} & \mathbf{F}_{\theta, g} \\
              \mathbf{F}_{\theta, g}^T & \mathbf{F}_g
       \end{bmatrix},
\end{align}
where $\mathbf{F}_{\theta,2}$, $\mathbf{F}_g$, and $\mathbf{F}_{\theta, g}$ are given by
\begin{subequations}\label{eq:F_theta-g}
\begin{align}
       \mathbf{F}_{\theta,2} =& 2\boldsymbol{\Lambda}_g \Re\{(\dot{\mathbf{A}}_{\gamma}\mathbf{J})^H \mathbf{R}_{\bar{n},2}^{-1} \dot{\mathbf{A}}_{\gamma}\mathbf{J}\} \boldsymbol{\Lambda}_g,\\
       \mathbf{F}_{g} =& 2 \Re\{\bar{\mathbf{A}}_{\gamma}^H \mathbf{R}_{\bar{n},2}^{-1} \bar{\mathbf{A}}_{\gamma}\},\\
       \mathbf{F}_{\theta, g} =& 2\boldsymbol{\Lambda}_g \Re\{(\dot{\mathbf{A}}_{\gamma}\mathbf{J})^H \mathbf{R}_{\bar{n},2}^{-1} \bar{\mathbf{A}}_{\gamma}\}.
\end{align}
\end{subequations}
To isolate the Fisher information exclusively pertaining to the elevation AoAs $\boldsymbol{\theta}$, we apply the Schur complement to eliminate the nuisance parameter $\mathbf{g}$, yielding the equivalent FIM as
\begin{align}\label{eq:F_eq_corr0}
       \mathbf{F}_{\mathrm{eq,corr}} = \mathbf{F}_{\theta,2} - \mathbf{F}_{\theta, g} \mathbf{F}_g^{-1} \mathbf{F}_{\theta, g}^T.
\end{align}
Substituting \eqref{eq:F_theta-g} into \eqref{eq:F_eq_corr0}, we rewrite $\mathbf{F}_{\mathrm{eq,corr}}$ in \eqref{eq:F_eq_corr}, which completes the proof.

\footnotesize
\bibliographystyle{IEEEtran}
\bibliography{reference}
\end{document}